\documentclass[a4paper,11pt]{article}
\usepackage{ifthen,calc}
\usepackage{amsmath,amssymb,amsthm,amstext,amsfonts,bm,mathrsfs,commath,comment}
\usepackage[scale=0.76]{geometry}
\usepackage{color,graphicx,pgf,tikz}
\usepackage{placeins}	
\usepackage{float}
\usepackage{subfigure}\usepackage{extarrows}

\usepackage{natbib}
\usepackage{enumitem}
\setlist{nosep}

\usepackage{setspace}
\allowdisplaybreaks
\usepackage{lineno}  

\newtheorem{defn}{Definition}
\newtheorem{thm}{Theorem}
\newtheorem{lem}{Lemma}

\newtheorem{prop}{Proposition}
\newtheorem{rmk}{Remark}
\newtheorem{exam}{Example}
\newtheorem{cro}{Corollary}

\newcommand{\bE}{\mathbb{E}}

\newcommand{\bR}{\mathbb{R}}
\newcommand{\bP}{\mathbb{P}}
\newcommand{\bZ}{\mathbb{Z}}
\newcommand{\cB}{\mathcal{B}}

\newcommand{\cI}{\mathcal{I}}

\newcommand{\cM}{\mathcal{M}}

\newcommand{\cS}{\mathcal{S}}

\newcommand{\cU}{\mathcal{U}}

\newcommand{\tens}[1]{%
  \mathbin{\mathop{\otimes}\limits_{#1}}%
}

\DeclareMathOperator{\rmd}{d\!}

\usepackage{hyperref}
\hypersetup{colorlinks=true,
	citecolor=blue,
	linkcolor=blue,
	urlcolor=blue,
	bookmarksopen=true,
	pdfstartview={XYZ null null 1.50},
	pdfpagelayout={SinglePage}
}

\usepackage{scalerel,stackengine}
\stackMath
\newcommand\ehat[1]{%
	\savestack{\tmpbox}{\stretchto{%
			\scaleto{%
				\scalerel*[\widthof{\ensuremath{#1}}]{\kern-.6pt\bigwedge\kern-.6pt}%
				{\rule[-\textheight/2]{1ex}{\textheight}}
			}{\textheight}%
		}{0.5ex}}%
	\stackon[1pt]{#1}{\tmpbox}%
}

\usetikzlibrary{calc,trees,positioning,arrows,chains,shapes,matrix,automata,patterns,shapes.geometric,decorations.pathreplacing,decorations.pathmorphing,decorations.markings,shapes.symbols}
\tikzset{>=latex}

\allowdisplaybreaks[4]

\begin{document}
\title{The equilibrium properties of direct strategy profiles \\ in games with many players\thanks{Parts of this paper were presented at the 21st SAET Conference (ANU, Canberra, Australia, July 16--22, 2022); the Asian Meeting of the Econometric Society (Beijing, China, June 30--July 2, 2023); the 9th PKU-NUS Annual International Conference on Quantitative Finance and Economics (Beijing, May 17--18, 2025); and the 24th SAET Conference (Ischia, Italy, June 29--July 5, 2025). We thank the participants at these events for their valuable comments and suggestions. Enxian Chen's research is supported by the National Natural Science Foundation of China (No. 72303115) and the Fundamental Research Funds for the Central Universities (No. 63222003). Bin Wu acknowledges support from the National Natural Science Foundation of China (No. 72203154 and No. 72450003).}}

\author{Enxian~Chen\thanks{School of Economics, Nankai University, Tianjin, 300071, China. E-mail: \href{mailto:chenenxian777@gmail.com}{chenenxian777@gmail.com}.}
\and
Bin~Wu\thanks{International School of Economics and Management, Capital University of Economics and Business, Fengtai District, Beijing, 100070, China. E-mail: \href{wubin@cueb.edu.cn}{wubin@cueb.edu.cn}.}
\and
Hanping Xu\thanks{Department of Mathematics, National University of Singapore, Block S17, 10 Lower Kent Ridge Road, 119076, Singapore. E-mail: \href{e0321212@u.nus.edu}{e0321212@u.nus.edu}.}
}

\date{This version: \today}

\maketitle

\begin{abstract}
This paper studies the equilibrium properties of the \emph{direct strategy profile} in large finite-player games. Each player in such a strategy profile simply adopts a strategy as she would have used in a symmetric equilibrium of an idealized large game. We show that, under a mild continuity condition, (i) direct strategy profiles constitute a convergent sequence of approximate equilibria as the number of players tends to infinity, and (ii) realizations of such strategy profiles also form a convergent sequence of (pure strategy) approximate equilibria with probability approaching one. Our findings provide a simple and decentralized approach for implementing equilibrium in large games, yielding outcomes that are asymptotically optimal both ex ante and ex post.



\end{abstract}

\clearpage
\tableofcontents
\setlength{\parskip}{4pt}
\clearpage

\section{Introduction}\label{sec-intro}
Games with a continuum of agents (referred to as large games) have been extensively studied in the literature. Typically, such games model strategic interactions involving a large but finite number of participants. By assuming a continuum of players where each player's unilateral deviation does not affect the aggregate action distribution, various desirable properties, such as the existence of pure strategy Nash equilibrium, have been established in these idealized large games. Despite significant progress in recent decades, the predominant focus in literature has been on the existence and convergence of Nash equilibria in large games. However, an important question remains unanswered: do players truly benefit from such an idealized process?  Specifically, while players may acknowledge that the Nash equilibrium of the idealized large game predicts a desirable social choice outcome, how do they behave to implement such an outcome in their real-life (finite) games? Furthermore, given an implementation strategy profile, do players have incentives to adhere to this designated strategy profile?

The first question seems to have a straightforward answer: each player could simply adopt the strategy they would select in an idealized large game. For instance, if a Nash equilibrium in the large game requires every player to choose a common strategy, then players in the real finite game can follow that equilibrium strategy. While this approach works in some situations, it often fails due to coordination issues arising from multiple individual-optimal choices.

To illustrate this, we consider a simple example of routing games. Imagine that many drivers traveling from an origin node $o$ to a terminal node $t$ via two paths, denoted as $a$ and $b$ respectively. The travel time on each path $p$ depends solely on the proportion of drivers using that path, denoted as $\tau(p).$ Path $b$  is a wider route where the travel time matches the proportion of drivers using it. In contrast, Path $a$  is narrower, causing congestion to build up more quickly, resulting in a travel time that is twice the proportion of drivers choosing it, i.e. $2\tau(a)$.

\begin{figure}[htb!]
	\centering
	\begin{tikzpicture}[>=stealth',shorten >=1pt,auto,node distance=2cm,semithick]
		\tikzstyle{every state}=[text=black]
		\node[state, red] at (0,0) (0) {$o$};
		\node[state, blue] at (8,0) (1) {$t$};
		\draw[->]   (0) -- (1, 1) -- node [above,sloped] {Path $a$} node [below,sloped] {Travel time is $2\tau(a)$} (7,1) -- (1);
		\draw[->]   (0) -- (1, -1) -- node [above,sloped] {Path $b$} node [below,sloped] {Travel time is $\tau(b)$} (7,-1) -- (1);
	\end{tikzpicture}
\end{figure}

This is a large but finite game, and to simplify the analysis, it is natural to consider its continuum-player counterpart. In the corresponding continuum game, there exists a unique Nash equilibrium action distribution in which one-third of the drivers choose the narrower Path~$a$, and the remaining two-thirds choose Path~$b$. In this equilibrium, both paths yield the same travel time of $\tfrac{2}{3}$, so each driver is indifferent between choosing $a$ and $b$. The resulting action distribution $\tau^* = (\tfrac{1}{3}, \tfrac{2}{3})$ represents the socially desirable outcome predicted by the continuum model. The challenge lies in implementing $\tau^*$ in the actual finite-player setting. Although $\tau^*$ specifies the target aggregate behavior, it does not uniquely pin down individual actions: both $a$ and $b$ are best responses. Without coordination, independently chosen actions may fail to reproduce $\tau^*$, leading to deviations from equilibrium. Achieving $\tau^*$ in pure strategies thus requires communication or centralized assignment---both of which become impractical in large, decentralized environments.

We address this issue in two steps. First, we observe that players in large games are typically distinguished by their characteristics, most notably by their payoff functions. For instance, drivers in a large traffic network may differ by vehicle type, with each type associated with distinct travel times. Similarly, buyers in an online marketplace can be grouped by their preferences over goods. In such environments, coordination becomes feasible if players base their strategies on observable characteristics, so that those with identical characteristics adopt the same strategy. This idea is formalized through the notion of symmetry. The first step of our approach is therefore to construct a symmetric equilibrium. Given an aggregate action distribution $\tau^*$ from a Nash equilibrium, we show that it can be implemented by a symmetric strategy profile $f^*$ satisfying two properties: (i) players with the same characteristic (i.e., the same payoff function) choose the same strategy in $f^*$, and (ii) the aggregate action distribution induced by $f^*$ equals $\tau^*$.

Second, we introduce the notion of a ``direct strategy profile'' for the finite-player game. In this strategy profile, each player selects a strategy based on her characteristic by adopting the strategy prescribed for that characteristic in the symmetric equilibrium $f^*$ of the continuum model. To illustrate, consider the routing game discussed earlier. Since all drivers share the same payoff function, they are identical in terms of characteristics and therefore adopt the same strategy in the symmetric equilibrium $f^*$. For instance, $f^*$ may prescribe a simple randomized strategy $\mu^*$ that assigns probability $\tfrac{1}{3}$ to Path $a$ and $\tfrac{2}{3}$ to Path $b$. In the corresponding direct strategy profile for the finite game, each player is assigned this strategy $\mu^*$.

Given the implementation via symmetrization and the direct strategy profile described above, a natural question arises: do players in the finite game have incentives to follow the direct strategy profile? This question concerns the equilibrium properties of such strategy profiles. In the routing game introduced earlier, the answer is affirmative: the direct strategy profile constitutes an approximate Nash equilibrium. Suppose there are $n$ drivers. Based on a combinatorial argument,\footnote{Each driver chooses Path~$a$ with probability $\tfrac{1}{3}$ and faces a random number of other drivers making the same choice. The expected travel time for choosing Path~$a$ is $2 \sum_{k = 0}^{n - 1} \left( \tfrac{1}{n} + \tfrac{k}{n} \right) \binom{n-1}{k} (\tfrac{1}{3})^k (\tfrac{2}{3})^{n - 1 - k} = \tfrac{2n + 4}{3n}$, while the expected travel time for Path~$b$ is $\sum_{k = 0}^{n - 1} \left( \tfrac{1}{n} + \tfrac{n - 1 - k}{n} \right) \binom{n-1}{k} (\tfrac{1}{3})^k (\tfrac{2}{3})^{n - 1 - k} = \tfrac{2n + 1}{3n}$.} 
the expected travel time for a driver choosing Path~$a$ is $\tfrac{2n + 4}{3n}$, while choosing Path~$b$ yields $\tfrac{2n + 1}{3n}$. Thus, the direct strategy profile forms a $\tfrac{1}{3n}$-Nash equilibrium\footnote{In an $\varepsilon$-Nash equilibrium, a large portion of players (more than $1-\varepsilon$) choose strategies that are within $\varepsilon$ of their optimal payoffs; see Definition~\ref{defn:approximateNE-finite} below.} in the $n$-player game. Moreover, as $n \to \infty$, the approximation error vanishes. This implies that following the direct strategy profile becomes asymptotically optimal in this routing game.

However, the asymptotic optimality of the direct strategy profile does not hold in general. Example~\ref{exam-Approximate} in Subsection~\ref{subsec-continuity} presents a large game $G$ with a Nash equilibrium $g$, for which the associated sequence of direct strategy profiles fails to converge. As the main result of this paper, we recover convergence by restricting attention to a class of large games that satisfy a continuity condition. In particular, Theorem~\ref{thm:main1} shows that for any large game admitting a convergent sequence of finite-player approximations, the corresponding sequence of direct strategy profiles forms a convergent sequence of approximate symmetric equilibria, provided a continuity condition is satisfied. This condition is relatively mild and holds in many large games studied in the literature, including those with finitely many characteristics and continuous large games with continuous equilibria.

Since a direct strategy profile typically involves randomized strategies, Theorem~\ref{thm:main1} establishes only ex ante asymptotic optimality. It is therefore natural to ask whether the direct strategy profile remains approximately optimal ex post, after the resolution of uncertainty. As the second main result of this paper, Theorem~\ref{thm:main2} shows that, under the same continuity condition, the sequence of realized direct strategy profiles forms a convergent sequence of (pure strategy) approximate equilibria, with probability approaching one. Together, the concept of the direct strategy profile and our two main theorems provide a complete answer to the motivating question: players in a large finite game benefit from the idealization process by playing the direct strategy profile.

Our main results offer a theoretical explanation for individual behavior in large strategic environments. For example, people increasingly use traffic apps, such as Google Maps, to select routes before traveling. In most traffic apps, a driver must input an origin point, a destination point, and the type of vehicle they drive (e.g., car, bus, or truck). The app then calculates the expected travel times for different feasible paths based on idealized models and historical traffic data. Finally, it recommends several optimal choices to the driver based on its calculations and the vehicle type. Typically, a driver will randomly select among the optimal choices recommended by the app. This behavior is consistent with the notion of a direct strategy profile. Given the large number of drivers, such decentralized choices become approximately optimal---both ex ante and, with high probability, ex post. Similar reasoning applies to other settings, such as buyer decisions in large online marketplaces.\footnote{For further examples of symmetric equilibrium in large economies, see \citet{Spiegler2006}, \citet{Spiegler2016}, and \citet{Tirole1988}.}

The remainder of the paper is organized as follows. Section~\ref{sec-model} introduces the models of large games and large finite-player games, along with the relevant equilibrium concepts. Our main results are presented in Section~\ref{sec-main results}. Section~\ref{sec-discussion} discusses related results and reviews the relevant literature. Finally, all proofs are collected in Section~\ref{sec-appendix}.


\section{The basic model}\label{sec-model}
In this section, we introduce some notation and basic definitions related to large games and large finite-player games. In such games, all players share a common compact action set, and each player's payoff depends on her own chosen action as well as the action distribution induced by all players' choices. The definitions in our model follow standard conventions in the literature; see, for example, \cite{KS2002}.

\subsection{Large games}\label{subsec-large}
A large game is defined as follows. Let $(I, \mathcal{I}, \lambda)$ be an atomless probability space representing the set of players.\footnote{Throughout this paper, we assume that all probability spaces are complete and countably additive.} Let $A$ be a compact metric space representing a common action set, equipped with the Borel $\sigma$-algebra $\mathcal{B}(A)$.\footnote{To simplify the analysis, we focus on large games with a common action space. All the main results in this paper can be generalized to settings where players have different feasible action sets. See Remark~\ref{rmk:action correspondence} for more details.} The set of probability measures on $A$ is denoted by $\mathcal{M}(A)$. Given an action profile for all players, the induced action distribution---which specifies the proportion of players choosing each action in $A$ (also referred to as the societal summary)---can be identified as an element of $\mathcal{M}(A)$. Each player's payoff is a bounded, continuous function on $A \times \mathcal{M}(A)$, reflecting continuous dependence on both her own action and the societal summary.

Let $\mathcal{U}_A$ be the space of bounded continuous functions on $A \times \mathcal{M}(A)$, endowed with the sup-norm topology and the corresponding Borel $\sigma$-algebra. In a general large game, a player's characteristic consists of her feasible action set and her payoff function $u_i$ (an element of $\mathcal{U}_A$). Since players share a common action set $A$ in our setting, $\mathcal{U}_A$ can be regarded as the characteristic space for all players.

A large game $G$ is a measurable mapping from $(I, \mathcal{I}, \lambda)$ to $\mathcal{U}_A$, assigning a payoff function to each player. A \emph{pure strategy profile} $f$ is a measurable function from $(I, \mathcal{I}, \lambda)$ to $A$. Let $\lambda f^{-1}$ denote the societal summary (also written as $s(f)$), which is the action distribution induced by $f$. Specifically, for any measurable subset $B \subseteq A$, $[s(f)](B)$ represents the proportion of players choosing actions in $B$.

A randomized strategy for each player $i$ is a probability measure $\mu \in \mathcal{M}(A)$. A \emph{randomized strategy profile} can be viewed as a measurable mapping $g: I \to \cM(A)$.\footnote{A measurable mapping $g: I \to \mathcal{M}(A)$ can also be viewed as a transition probability $g: I \times \mathcal{B}(A) \to [0,1]$ such that: (i) For every $B \in \mathcal{B}(A)$, $g(\cdot; B)$ is measurable; (ii) For $\lambda$-almost all $i \in I$, $g(i; \cdot) \in \mathcal{M}(A)$.} Note that every pure strategy profile $f$ naturally induces a randomized strategy profile $g^f$ defined by $g^f(i) = \delta_{f(i)}$ for each player $i \in I$.\footnote{Here $\delta_{f(i)}$ denotes the Dirac probability measure that assigns probability one to $\{f(i)\}$.} Given a randomized strategy profile $g$, the societal summary $s(g)$ is modeled as the average action distribution across all players, that is, $s(g) = \int_I g(i) \, \mathrm{d}\lambda(i) \in \mathcal{M}(A)$. When $A$ is a finite set of actions, $s(g)$ can be identified with a point in the simplex in $\mathbb{R}^{|A|}$. In the case of an infinite action set, $s(g)$ satisfies $[s(g)](B) = \int_I g(i,B) \, \mathrm{d}\lambda(i)$ for every measurable subset $B \subseteq A$. Moreover, a randomized strategy profile $g$ is said to be \emph{symmetric} if for any two players $i$ and $i'$, we have $g(i) = g(i')$ whenever $G(i) = G(i')$, i.e., players with the same characteristic (payoff function) use the same strategy.

We now present the equilibrium concepts for large games. The formal definition of a randomized strategy Nash equilibrium is as follows.

\begin{defn}[Randomized strategy Nash equilibrium]
\label{defn:randNE}
\rm
A randomized strategy profile $g \colon I \rightarrow \cM(A)$ is said to be a \emph{randomized strategy Nash equilibrium} if for almost all $i \in I$,
$$ \int_{A} u_i \bigl( a, s(g) \bigr) g(i, \dif a)  \ge  \int_{A} u_i \bigl( a, s(g) \bigr) \dif\mu(a) \text{ for all } \mu \in \cM(A). $$
\end{defn}

Therefore, a randomized strategy profile $g$ is a \emph{randomized strategy Nash equilibrium} if it is optimal for almost every player with respect to the societal summary $s(g)$, in terms of expected payoff. In the case of a pure strategy profile $f$, the societal summary $s(f)$ coincides with the action distribution $\lambda f^{-1}$. This leads to the following definition of a pure strategy Nash equilibrium.

\begin{defn}[Pure strategy Nash equilibrium]
\label{defn:pureNE}
\rm
A pure strategy profile $f \colon I \rightarrow A$ is said to be a \emph{pure strategy Nash equilibrium} if, for all $i\in I,$
$$u_{i}\bigl(f(i),  \lambda f^{-1}\bigr) \ge u_{i}(a, \lambda f^{-1}) \text{ for all } a \in A.$$
\end{defn}

A randomized strategy Nash equilibrium always exists in large games. By contrast, the existence of a pure strategy Nash equilibrium generally requires additional conditions: it is guaranteed in large games with finite action sets (\cite{S1973}) or when the player space satisfies certain saturation condition (\cite{HSS2017}).\footnote{See also \cite{QY2014} for counterexamples and further discussion.}


\subsection{Large finite-player games}\label{subsec-finite}
In this subsection, we introduce a class of large finite-player games. Let $(I_n, \mathcal{I}_n, \lambda_n)$ be a finite probability space representing the set of players. We assume that $|I_n| = n$, that $\mathcal{I}_n$ is the power set of $I_n$, and that $\lambda_n(i) = \tfrac{1}{n}$ for each $i \in I_n$.\footnote{The assumption that each player has equal weight is not essential, but we adopt it to simplify the model.} For simplicity, we assume that all players share a common action space $A$.

Similarly, each player's payoff function depends on her own choice and the probability distribution on $A$ that is induced from the action profile (i.e., the societal summary). Clearly, the set of such action distributions is a subset of $\cM(A)$ that is denoted by
$$
D^n = \Bigl\{   \tau \in \cM(A) \;\Big\vert\; \tau =  \sum\limits_{j \in I_n} \lambda_n(j) \delta_{ a_j }  \text{  where } a_j \in A \text{ for all } j \in I_n          \Bigr\},
$$
Player $i$'s payoff function is then given by a bounded continuous function $u_i^n \colon A \times \mathcal{M}(A) \to \mathbb{R}$, clearly, $u_i^n \in \mathcal{U}_A$. Thus, a large finite-player game $G_n$ can be represented as a mapping from $I_n$ to $\mathcal{U}_A$, with $G_n(i) = u_i^n$ for all $i \in I_n$.

In this finite-player game, a \emph{pure strategy profile} $f^n$ is a mapping from $( I_n , \cI_n, \lambda_n)$ to $A$. Hence, given a pure strategy profile $f^n$, the payoff function for player $i$ is
$$
u_i^n(f^n) = u_i^n \Bigl( f^n(i), \sum\limits_{j \in I_n } \lambda_n(j) \delta_{f^n(j)} \Bigr),
$$
here we slightly abuse the notation $u_i^n(f^n)$ to denote player $i$'s payoff given the strategy profile $f^n$. Similarly, a randomized strategy is a probability distribution $\mu \in \cM(A)$. A \emph{randomized strategy profile} $g^n$ is a mapping from  $( I_n , \cI_n, \lambda_n)$ to $\cM(A)$. Thus, given a randomized strategy profile $g^n$, player $i$'s (expected) payoff is given by
$$
u_i^n(g^n) = \int_{A^n} u_i^n \Bigl(a_i, \sum\limits_{j \in I_n } \lambda_n(j) \delta_{a_j} \Bigr) \tens{j \in I_n} g^n(j, {\rm{d}} a_j),
$$
where $\tens{ j \in I_n} g^n(j,  {\rm{d}} a_j)$ is the product probability measure on the product space $A^n$. The societal summary induced by $g^n$ is $s(g^n) = \int_{I_n} g^n(i) \mathrm{d} \lambda_n(i)$. Moreover, a randomized strategy profile $g^n$ is said to be \emph{symmetric} if for any two players $i$ and $i'$, $g^n(i) = g^n(i')$ whenever $G_n(i) = G_n(i')$. Finally, we state the definitions of  randomized strategy Nash equilibrium and $\varepsilon$-Nash equilibrium as follows.

\begin{defn}[Randomized strategy Nash equilibrium]
	\label{defn:randNE-finite}
	\rm
	A randomized strategy profile $g^n \colon I_n \rightarrow \cM(A)$ is said to be a \emph{randomized strategy Nash equilibrium} if for all $i \in I_n$,
	$$ U_i^n(g^n) \ge  U_i^n( \mu  ,g^n_{- i}  ) \text{ for all } \mu \in \cM(A), $$
	where $(\mu, g^n_{-i})$ represents the randomized strategy profile such that player $i$ plays the randomized strategy $\mu$, and player $j$ plays the randomized strategy $g^n(j)$ for all $j \in I_n \backslash \{i\}$.
\end{defn}

\begin{defn}[$\varepsilon$-Nash equilibrium]
 	\label{defn:approximateNE-finite}
     \rm
 	For any  $\varepsilon \ge 0$, a randomized strategy profile $g^n \colon I_n \to \cM(A)$ is said to be an $\varepsilon$-Nash equilibrium if there exists a subset of players $I_n^{\varepsilon} \subseteq I_n$ such that $\lambda_n( I_n^{\varepsilon}) \ge 1 - \varepsilon$ and for all $i \in I_n^{\varepsilon}$,
 	\[
 	U_i^n(g^n)  \ge  U_i^n (\mu,g^n_{-i} ) - \varepsilon \text{ for all } \mu \in \cM(A).
 	\]
\end{defn}

Thus, in an $\varepsilon$-Nash equilibrium, most players choose strategies that are within $\varepsilon$ of their best responses, and only a small portion of players (no more than $\varepsilon$) may obtain higher than $\varepsilon$ by deviation. Clearly, a Nash equilibrium is also an $\varepsilon$-Nash equilibrium ($\varepsilon = 0$).

Throughout the rest of this paper, a Nash equilibrium always refers to a randomized strategy Nash equilibrium, and an $\varepsilon$-Nash equilibrium is also called an approximate Nash equilibrium.


\section{Main results}\label{sec-main results}
In this section, we present the main results. Let $G_n$ be a large finite-player game with $n$ players. To simplify the equilibrium analysis, we assume that the players in $G_n$ use a large game $G$---featuring a continuum of players---as an approximation. Specifically, we assume that the characteristic information (i.e., the payoff functions) of the players in $G_n$ is contained within that of the large game $G$. Formally, this means $G_n(I_n) \subset \mathrm{supp}\, \lambda G^{-1}$.\footnote{Since $G$ is a measurable mapping from the player space $(I, \mathcal{I}, \lambda)$ to the characteristic space $\mathcal{U}_A$, the measure $\lambda G^{-1}$ represents the distribution of characteristics among players. Its support, denoted by $\mathrm{supp}\, \lambda G^{-1}$, is the smallest closed subset of $\mathcal{U}_A$ that has full measure.} As we will show, this assumption ensures that the players in the finite-player game $G_n$ can directly implement an equilibrium outcome of the large game $G$.

Suppose that $\tau^*$ is the societal summary (i.e., the action distribution) induced by a Nash equilibrium $g$ of the large game $G$. This equilibrium action distribution $\tau^*$ is considered a desirable social choice outcome for the players in the finite-player game $G_n$. As discussed in Section~\ref{sec-intro}, when attempting to implement the outcome $\tau^*$, each player may face multiple optimal responses, creating a coordination problem. To resolve this issue, we adopt a process to symmetrize the equilibrium action distribution. Specifically, there exists a symmetric Nash equilibrium $\widetilde{g}$ of the large game $G$ such that its societal summary also equals $\tau^*$; that is, $s(\widetilde{g}) = \tau^*$. This symmetrization result relies on the existence of an auxiliary mapping, defined as follows.

\begin{defn}[Auxiliary mapping]
 	\label{defn:map}
     \rm
Given an equilibrium action distribution $\tau^*$, the \emph{auxiliary mapping} associated with $\tau^*$ is a function $\overline{g} \colon \mathcal{U}_A \to \mathcal{M}(A)$ such that the composition $\widetilde{g} := \overline{g} \circ G \colon I \to \mathcal{M}(A)$ is a Nash equilibrium of the large game $G$ and satisfies $s(\widetilde{g}) = \tau^*$.
\end{defn}

\begin{figure}[htb!]
	\centering
\begin{tikzpicture}
\def\a{1.5} \def\b{2}
\path
(-\a,0) node (A) {$I$}      
(\a,0) node (B) {$\mathcal{M}(A)$}
(0,-\b) node[align=center] (C) {$\mathcal{U}_{A}$};
\begin{scope}[nodes={midway,scale=.75}]
\draw[->] (A)--(B) node[above]{$\widetilde{g}$};
\draw[->] (A)--(C.120) node[left]{$G$};
\draw[dashed,->] (C.60)--(B) node[right]{$\overline{g}$};
\end{scope}
\end{tikzpicture}  
\caption{Construction of a symmetric equilibrium via the auxiliary mapping} \label{fig1}
\end{figure}

\noindent
Note that the auxiliary mapping $\overline{g}$ assigns a strategy to each possible player characteristic. As a result, in the composed profile $\widetilde{g} = \overline{g} \circ G$, players who share the same characteristic adopt the same strategy. Lemma~\ref{lem:auxiliary} below guarantees the existence of such a symmetric equilibrium $\widetilde{g}$ whose societal summary coincides with $\tau^*$.

\begin{lem}\label{lem:auxiliary}
For any equilibrium action distribution $\tau^*$ of a large game $G$, the auxiliary mapping $\overline{g}$ always exists.
\end{lem}

Lemma~\ref{lem:auxiliary} follows as a corollary of \citet[Theorem 2(ii)]{SSY2020}, which relies on the technique of disintegration to establish the existence of $\overline{g}$. Since our model focuses on games with a common action set, we provide a shorter and more direct proof in Subsection~\ref{subsec-proof auxiliary}. Notably, playing a symmetric strategy profile $\widetilde{g}$ resolves coordination issues, as players need only respond based on their characteristics. Moreover, since the societal summary of $\widetilde{g}$ coincides with $\tau^*$, players can coordinate on the symmetric equilibrium $\widetilde{g}$ to implement the desired equilibrium outcome $\tau^*$.

Each player $i$ in the finite-player game $G_n$ is associated with an \emph{imagined partner} $i'$ in the large game $G$ who shares the same characteristic; that is, $G_n(i) = G(i')$.\footnote{This relies on the assumption introduced at the beginning of this section: $G_n(I_n) \subset \mathrm{supp}\, \lambda G^{-1}$.} A natural way for players in $G_n$ to implement the desired societal summary $\tau^*$ is to imitate the strategy used by their imagined partners in $G$. Specifically, player $i$ adopts the strategy $\widetilde{g}(i') = \overline{g} \circ G(i') = \overline{g} \circ G_n(i)$. The resulting strategy profile in $G_n$ is referred to as the direct strategy profile, formally defined below.

\begin{defn}[Direct strategy profile]
 	\label{defn:direct}
     \rm
Given a large finite-player game $G_n$, a large game $G$, and an associated auxiliary mapping $\overline{g}$, the \emph{direct strategy profile} of $G_n$ is defined by $g^n = \overline{g} \circ G_n$.
\end{defn}
  
\begin{figure}[htb!]
	\centering 
  \begin{tikzpicture}
\def\a{1.5} \def\b{2}
\path
(-\a,0) node (A) {$I_n$}      
(\a,0) node (B) {$\mathcal{M}(A)$}
(0,-\b) node[align=center] (C) {$\mathcal{U}_{A}$};
\begin{scope}[nodes={midway,scale=.75}]
\draw[dashed, ->] (A)--(B) node[above]{$g^n$};
\draw[->] (A)--(C.120) node[left]{$G_n$};
\draw[->] (C.60)--(B) node[right]{$\overline{g}$};
\end{scope}
\end{tikzpicture}   
\caption{Direct strategy profile} \label{fig2}
\end{figure}

\noindent
Observe that the direct strategy profile $g^n$ is symmetric, as players with the same characteristic choose the same strategy, determined by the auxiliary mapping $\overline{g}$. In this strategy profile, each player adopts the strategy that her imagined partner would use in the symmetric equilibrium $\widetilde{g}$ of the large game $G$.

We now revisit the motivating example from Section~\ref{sec-intro} to illustrate the roles of $\widetilde{g}$, $\overline{g}$, and $g^n$. In that example, $\widetilde{g}(i) = \tfrac{1}{3} \delta_a + \tfrac{2}{3} \delta_b$ for every player $i$ in the large routing game. Since all players share the same payoff function $u$, the image $G(I)$ is the singleton $\{u\}$, and thus $\overline{g}(u) = \tfrac{1}{3} \delta_a + \tfrac{2}{3} \delta_b$. Consequently, for any finite routing game $G_n$ with $n$ drivers, the corresponding direct strategy profile satisfies $g^n(i) = \tfrac{1}{3} \delta_a + \tfrac{2}{3} \delta_b$ for each driver $i \in I_n$.

Our main theorems examine the equilibrium properties of the direct strategy profile $g^n$. As discussed in Section~\ref{sec-intro}, $g^n$ is generally not an exact Nash equilibrium; instead, it can be regarded as an $\varepsilon_n$-Nash equilibrium for some $\varepsilon_n \geq 0$. Varying $n$, we obtain a sequence of direct strategy profiles $\{g^n\}_{n \in \mathbb{Z}_{+}}$ corresponding to the sequence of finite-player games $\{G_n\}_{n \in \mathbb{Z}_{+}}$. This raises a natural question: does the sequence $\{g^n\}_{n \in \mathbb{Z}_{+}}$ converge? More specifically, does $\varepsilon_n$ tend to zero as the games $\{G_n\}_{n \in \mathbb{Z}_{+}}$ converge to $G$? To answer this, we first introduce the following concept regarding the convergence of games.

\begin{defn}[Convergence of games]
 	\label{defn:appro}
     \rm
A sequence of finite-player games $\{G_n\}_{n \in \mathbb{Z}_{+}}$ converges to a large game $G$ if the corresponding distributions of player characteristics converge weakly; that is,
$$
\lambda_n \circ G_n^{-1} \xrightarrow{w} \lambda \circ G^{-1} \quad \text{on } \mathcal{U}_A.
$$
\end{defn}

Definition~\ref{defn:appro} formalizes a widely used notion of convergence in the study of large games: the weak convergence of the distributions of player characteristics. Under this notion, the sequence of characteristic distributions in the finite-player games $\{G_n\}_{n \in \mathbb{Z}_{+}}$ converges to the limiting distribution induced by the large game $G$. The game-theoretic interpretation is that, as the number of players grows, individual heterogeneity becomes negligible, and the aggregate structure of the game stabilizes. This convergence justifies using large games as approximations for analyzing strategic behavior in large but finite populations.

We are now ready to present our first main result. Theorem~\ref{thm:main1} shows that if the auxiliary mapping $\overline{g}$ is almost everywhere continuous on the support of the characteristic distribution (i.e., $\mathrm{supp}\,\lambda G^{-1}$), then the sequence of direct strategy profiles $\{g^n\}_{n \in \mathbb{Z}_{+}}$---each constructed by applying $\overline{g}$ to the characteristics in $G_n$---forms a convergent sequence of approximate Nash equilibria.

\begin{thm}
	\label{thm:main1}
	Let $\{G_n\}_{n \in \mathbb{Z}_{+}}$ be a sequence of finite-player games that converges to a large game $G$. Let $\overline{g}$ be an auxiliary mapping associated with $G$, and define $\{g^n\}_{n \in \mathbb{Z}_{+}}$ as the sequence of direct strategy profiles induced by $\overline{g}$. If $\overline{g}$ is almost everywhere continuous on $\mathrm{supp}\,\lambda G^{-1}$, then there exists a sequence $\{\varepsilon_n\}_{n \in \mathbb{Z}_{+}}$ with $\varepsilon_n \to 0$ such that each $g^n$ is an $\varepsilon_n$-Nash equilibrium of $G_n$.
\end{thm}

Example~\ref{exam-Approximate} in Subsection~\ref{subsec-continuity} below illustrates that the continuity condition imposed on the auxiliary mapping $\overline{g}$ is necessary for the conclusion of Theorem~\ref{thm:main1} to hold. Nonetheless, this requirement is relatively mild and satisfied by many classes of large games studied in the literature. A particularly important case occurs when the large game features only finitely many player characteristics---that is, when $G(I)$ is a finite subset of $\mathcal{U}_A$. In such settings, $\overline{g}$ is trivially almost everywhere continuous, since any function defined on a finite set is continuous on its support. Consequently, Theorem~\ref{thm:main1} applies directly, and as shown in Corollary~\ref{coro:finite type}, the sequence of direct strategy profiles constructed from the corresponding finite-player games forms a convergent sequence of approximate symmetric equilibria.

\begin{cro}
	\label{coro:finite type}
	Let $G \colon (I, \cI, \lambda) \to \cU_A$ be a large game such that $G(I)$ is a finite subset of $\cU_A$, and let $\{G_n\}_{n \in \mathbb{Z}_+}$ be a sequence of finite games converging to $G$. Suppose $\{g^n\}_{n \in \mathbb{Z}_+}$ is the sequence of direct strategy profiles induced by the auxiliary mapping $\overline{g}$. Then, there exists a sequence $\{\varepsilon_n\}_{n \in \mathbb{Z}_+}$ with $\varepsilon_n \to 0$ such that each $g^n$ is an $\varepsilon_n$-Nash equilibrium of $G_n$.
\end{cro}

This result has practical significance. In many real-world applications, the number of distinct player characteristics is limited, effectively making the characteristic space finite. Therefore, Theorem~\ref{thm:main1} together with Corollary~\ref{coro:finite type} imply that, in large finite-player games, it is asymptotically optimal for players to adopt the strategies played by their imagined partners in the corresponding large game.

\begin{rmk}\label{rmk:action correspondence}
\rm
Theorem~\ref{thm:main1} can be extended to games in which players may have different feasible action spaces---i.e., each player's characteristic consists of both a feasible action set and a payoff function. In this more general setting, the notions of auxiliary mapping and direct strategy profile can be defined in a similar manner. Moreover, if the auxiliary mapping is almost everywhere continuous on $\mathrm{supp}\, \lambda G^{-1}$, then the sequence of induced direct strategy profiles converges.
\end{rmk}

\begin{rmk}\label{rmk:continuous}
\rm
Beyond the setting of games with finitely many characteristics, the continuity condition on $\bar g$ is also satisfied in many games with infinitely many characteristics. For instance, consider the case where both the characteristic mapping $G \colon I \to \cU_A$ and the strategy profile $\widetilde g \colon I \to \cM(A)$ are continuous.\footnote{Here, $\cU_A$ is endowed with the sup-norm topology, $I$ is equipped with the topology induced by $G$, and $\cM(A)$ is endowed with the weak topology.} Then, as illustrated in Figure~\ref{fig1}, the auxiliary mapping $\overline g$ is also continuous. For related discussions on continuous Nash equilibria in games with continuous player spaces, see \cite{Kim1997}.\footnote{\cite{Kim1997} adopts the argmax topology on the characteristic space $\cU_A$, which differs from the sup-norm topology used here.}
\end{rmk}

We conclude this section by presenting our second main result, which examines the equilibrium property of direct strategy profiles after the resolution of uncertainty.  Given a sequence $\{g^n\}_{n \in \mathbb{Z}_+}$, where each $g^n$ is a randomized strategy profile in the finite-player game $G_n$, standard probability theory ensures the existence of a probability space $(\Omega, \Sigma, \mathbb{P})$ and a sequence of measurable mappings $\{x^n \colon I_n \times \Omega \to A\}_{n \in \mathbb{Z}_+}$ such that:
\begin{itemize}
    \item[(i)] For each $i \in I_n$ and $n \in \mathbb{Z}_+$, the distribution of the random variable $x^n_i := x^n(i, \cdot)$ coincides with $g^n(i)$;
    \item[(ii)] For each $n \in \mathbb{Z}_+$, the family of random variables $\{x^n_i\}_{i \in I_n}$ are pairwise independent.
\end{itemize}
For each \(\omega \in \Omega\), the mapping \(x^n(\cdot, \omega) \colon I_n \to A\) defines a pure strategy profile, called a \emph{realization} of \(g^n\).

For each $n \in \bZ_+ $ and $\varepsilon \ge 0$, we consider the subset $\Omega^n_\varepsilon \subseteq \Omega$ consisting of realizations that are $\varepsilon$-Nash equilibria in pure strategies:
$$
\Omega^{n}_\varepsilon = \left\{\omega \in \Omega \,\middle|\, x^n(\cdot, \omega) \text{ is an } \varepsilon\text{-pure strategy Nash equilibrium of } G_n \right\}.
$$
We are now ready to state the second main result, which shows that realizations of the direct strategy profiles form a convergent sequence of approximate pure strategy equilibria with probability approaching one.

\begin{thm}
	\label{thm:main2}
	Given a sequence of finite-player game $\{G_n\}_{n \in \bZ_+}$ that converges to a large game $G$, let $\overline{g}$ be an auxiliary mapping of $G$. Define $\{g^n\}_{n \in \bZ_{+}}$ as the sequence of direct strategy profiles induced by $\overline{g}$. If $\overline{g}$ is almost everywhere continuous on $\mathrm{supp } \lambda G^{-1}$, then for any $\varepsilon > 0$, 
$$\lim\limits_{n \to \infty} \bP(\Omega^{n}_\varepsilon) = 1.$$
\end{thm}

As in Theorem~\ref{thm:main1}, the continuity assumption on $\overline{g}$ in Theorem~\ref{thm:main2} is also essential, as illustrated by Example~\ref{exam-Approximate} in Subsection~\ref{subsec-continuity}. Together, Theorems~\ref{thm:main1} and \ref{thm:main2} establish that, for players in a large finite-player game, adopting the direct strategy profile is asymptotically optimal---not only ex ante, in expectation, but also ex post, with high probability.


\section{Discussion}\label{sec-discussion}
This section is divided into three parts. Subsection~\ref{subsec-continuity} presents a counterexample demonstrating that the continuity condition imposed in Theorems~\ref{thm:main1} and~\ref{thm:main2} is essential. Subsection~\ref{subsec-symmetric pure} investigates direct strategy profiles in pure strategies, showing that, under certain conditions, such strategy profiles exist and form a convergent sequence. Finally, Subsection~\ref{subsec-literature} reviews related contributions from the existing literature.

\subsection{Failure of convergence without the continuity condition}\label{subsec-continuity}

In this subsection, we demonstrate that the continuity assumption on the auxiliary mapping $\overline{g}$ in Theorems~\ref{thm:main1} and~\ref{thm:main2} is essential. We construct a counterexample showing that, if $\overline{g}$ fails to be continuous, then the resulting sequence of direct strategy profiles may fail to form a convergent sequence of approximate symmetric equilibria---specifically, the associated approximation errors $\{\varepsilon_n\}_{n \in \mathbb{Z}_+}$ may not converge to zero as $n \to \infty$.

\begin{exam}\label{exam-Approximate}
\rm
Consider a large game $G$ where the player space is the Lebesgue unit interval $(I, \mathcal{I}, \lambda)$, and all players share a common action set $A = \{0, 1\}$. Each player $i \in I$ has a payoff function defined by
$$
u_i(a, \tau) = i + \bigl(a - \tau(1)\bigr)^2,
$$
so the large game $G$ can be represented by the mapping $G(i) = u_i$ for all $i \in I$.

Let $\mathbb{Q}$ denote the set of rational numbers in $\mathbb{R}$, and define a randomized strategy profile $g \colon I \to \cM(A)$ as follows:
\[
g(i) =
\begin{cases}
\delta_1 & \text{if } i \in \mathbb{Q} \cap [0,1], \\
i \delta_0 + (1 - i) \delta_1 & \text{if } i \in [0,1] \setminus \mathbb{Q}.
\end{cases}
\]

Since all players have distinct payoff functions, $g$ is a symmetric strategy profile. The associated societal summary is $s(g) = \int_I g(i) \, d\lambda(i) = \tfrac{1}{2} \delta_0 + \tfrac{1}{2} \delta_1.$ It follows that $u_i\bigl(0, s(g)\bigr) = u_i\bigl(1, s(g)\bigr)$ for every $i \in I$, so $g$ is a Nash equilibrium. Thus, $g$ is a symmetric Nash equilibrium, and its symmetrized profile $\widetilde{g}$ coincides with $g$, i.e., $g = \widetilde{g} = \overline{g} \circ G$.

Observe that the characteristic mapping $G \colon I \to \cU_A$ is everywhere continuous, while the strategy profile $g \colon I \to \cM(A)$ is nowhere continuous. Hence, the auxiliary mapping $\overline{g} \colon \cU_A \to \cM(A)$ is also nowhere continuous.

Next, define a sequence of finite-player games $\{G_n\}_{n \in \mathbb{Z}_+}$ that converges to $G$. For each $n \in \mathbb{Z}_+$, let the player set be $I_n = \{k/n : k = 1, \ldots, n\}$, equipped with the uniform measure $\lambda_n(i) = 1/n$ for all $i \in I_n$. Each player has the same action set $A = \{0, 1\}$, and for each $i \in I_n$, define the payoff function by $u_i^n(a, \tau) = u_i(a, \tau)$ for all $a \in A$ and $\tau \in \cM(A)$, so $G_n(i) = u_i$ for all $i \in I_n$.

For each $n$, the direct strategy profile induced by $\overline{g}$ is given by $g^n(i) = \overline{g}\bigl(G_n(i)\bigr) = \overline{g}\bigl(G(i)\bigr) = g(i)$. Since all elements of $I_n$ are rational numbers, we have $g^n(i) = \delta_1$ for every $i \in I_n$. However, this sequence of direct strategy profiles does not form a convergent sequence of approximate Nash equilibria. Indeed, for each $i \in I_n$, the best deviation from $\delta_1$ yields a payoff gain of at least $(1 - \tfrac{1}{n})^2$. Therefore, if $g^n$ is regarded as an $\varepsilon_n$-Nash equilibrium of $G_n$, it must be that $\varepsilon_n \ge (1 - \tfrac{1}{n})^2$, which does not converge to zero as $n \to \infty$.

The above reasoning shows that the continuity condition in Theorem~\ref{thm:main1} is indispensable. Moreover, since $\widetilde{g}(i)$ is a pure strategy for each $i \in I_n$, each $g^n(i)$ is likewise a pure strategy. As a result, the sequence of realizations of direct strategy profiles $\{g^n\}_{n \in \mathbb{Z}_{+}}$ also fails to form a convergent sequence of approximate Nash equilibria. This confirms that the continuity assumption is equally essential for Theorem~\ref{thm:main2}.
\end{exam}


\subsection{Direct strategy profiles in pure strategies}\label{subsec-symmetric pure}
A direct strategy profile induced by an auxiliary mapping generally involves randomized strategies. A natural question is whether such a strategy profile can be constructed in pure strategies. A notable special case arises when the symmetric equilibrium $\widetilde g$ of the large game is itself a pure strategy profile. In that case, the auxiliary mapping $\overline g$ assigns pure strategies to player characteristics, and the resulting direct strategy profile $g^n = \overline g \circ G_n$ is also composed of pure strategies.

To investigate the existence of symmetric pure strategy equilibria in large games, we adopt the atomless condition introduced by \citet{M1984}, which guarantees such equilibria exist when the action set is finite.

\begin{defn}[Atomless condition]
 	\label{defn:atomless}
     \rm
A large game $G \colon I \to \cU_A$ is said to satisfy the \emph{atomless condition} if the induced distribution $\lambda G^{-1}$ on $\cU_A$ is atomless.  
\end{defn}

\citet[Theorem~2]{M1984} shows that a pure strategy symmetric equilibrium exists in any large game with a finite action set that satisfies the atomless condition.\footnote{\cite{KS1995a, KS1995b} generalized this existence result to large games with countable action sets by establishing a pure symmetrization theorem under the atomless condition.} Building on this result and Theorem~\ref{thm:main1}, we obtain the following proposition. Throughout this subsection, we restrict attention to games with finite action sets.

\begin{prop}
	\label{prop:pure}
	Given a sequence of finite game $\{G_n\}_{n \in \bZ_+}$ that converges to a large game $G$ satisfying the atomless condition, let $\overline{g}$ be an auxiliary mapping of $G$ induced from a pure strategy symmetric equilibrium $\widetilde g$. Define $\{g^n\}_{n \in \bZ_{+}}$ as the sequence of direct strategy profiles induced by $\overline{g}$. If $\overline{g}$ is almost everywhere continuous on $\mathrm{supp } \lambda G^{-1}$, then there exists a sequence of real numbers $\{\varepsilon_n\}_{n \in \bZ_+}$ such that each $g^n$ is an $\varepsilon_n$-Nash equilibrium in pure strategy, and $\varepsilon_n \to 0$ as $n$ approaches infinity.
\end{prop}

The continuity condition on $\overline g$ remains essential in Proposition~\ref{prop:pure}. Indeed, Example~\ref{exam-Approximate} satisfies the atomless condition, yet the induced sequence of direct strategy profiles fails to converge in the absence of continuity.

Proposition~\ref{prop:pure} can also be extended to large games with more general action spaces. For example, in a large game with a countable action set, one can show---building on the results of \citet{KS1995a, KS1995b}---that if the auxiliary mapping is almost everywhere continuous, then the sequence of direct strategy profiles still forms a convergent sequence of approximate equilibria.


\subsection{Related literature}\label{subsec-literature}
Substantial progress has been made in the theory of large games since the foundational works of \citet{S1973} and \citet{M1984}. More recent contributions include \citet{Kalai2004}, \citet{Sun2006}, \citet{KRSY2013}, \citet{QY2014}, \citet{Yu2014}, \citet{DK2015}, \citet{HSS2017}, \citet{KRYZ17}, \citet{KS2018}, \citet{KRQS2020}, \citet{SSY2020}, \citet{CP2022}, \citet{CQSS2022}, \citet{Hellwig2022}, \citet{W2022}, \citet{Yang2022}, and \citet{CWX2025}. Significant advances has also been made in the analysis of large economies; see, for instance, \citet{Yannelis2009}, \citet{HY2016}, and \citet{ADKU2022a, ADKU2022b}. This section reviews literature most directly related to our analysis.

This paper builds on prior work on symmetric equilibria in large games, notably \citet{M1984}, \citet{KS1995a, KS1995b}, and \citet{SSY2020}. \citet{M1984} established the existence of pure strategy symmetric equilibria under the atomless condition and finite action sets. \citet{KS1995a, KS1995b} extended this result to games with countable action sets. More recently, \cite{SSY2020} considered randomized strategy symmetric equilibria and showed the existence of such equilibria in general large games. These results provide the theoretical foundation for our analysis. In contrast to most of this literature, which focuses on a continuum of players, we examine games with a large but finite population.

Our approach is also related to \citet{KS2018}, who studied incomplete-information repeated games with a large but finite number of players. To address analytical complexity, they introduced the notion of a (symmetric) imagined-continuum equilibrium, in which the societal summary---referred to as the ``macro-strategy''---was unaffected by individual deviations. They established the existence of a symmetric Markov equilibrium in the imagined-continuum game and showed that it induced an approximate equilibrium in the finite-player game under a Lipschitz continuity condition on the outcome-generating function. While both papers examined symmetric equilibria in large finite games, key differences emerged. \citet{KS2018} adopted symmetry as a simplifying assumption by considering homogeneous players with common payoffs and priors. This ensured that the same strategy applied in both continuum and finite-player settings, rendering the continuity condition in Theorem~\ref{thm:main1} trivially satisfied. In contrast, our model allows for heterogeneous players, rendering the continuity condition on the auxiliary mapping essential (see Example~\ref{exam-Approximate}). Moreover, whereas \citet{KS2018} assumed a finite action set and Lipschitz continuous payoffs, our framework accommodates general action spaces without these restrictions.

Theorem~\ref{thm:main2} relates to \citet{Kalai2004}, who studied the approximate ex post stability of Nash equilibria in Bayesian games with a large but finite number of players. He showed that if payoff functions satisfy certain equicontinuity conditions, then realizations of exact equilibria in the large finite game are approximately stable with high probability. The key distinction is that Theorem~\ref{thm:main2} considers a Nash equilibrium of the limiting large game $G$, which may not correspond to an exact Nash equilibrium in each finite-player game $G_n$. Moreover, as illustrated in Example~\ref{exam-Approximate}, the continuity assumption on $\overline{g}$ in Theorem~\ref{thm:main2} is essential.

Finally, the paper contributes to the literature on the convergence of equilibria in large games. Earlier work (e.g., \citet{Gr1984}, \citet{H1988}) examined whether equilibrium correspondences are upper hemicontinuous as the number of players grows. More recent developments include \citet{KRSY2013}, \citet{QY2014}, and \citet{HSS2017}. These papers focus primarily on pure strategies. We extend existing results by analyzing the convergence of randomized strategies, a technically distinct and less developed area in the study of large games.

\section{Appendix}\label{sec-appendix}

\subsection{Technical preparations}\label{subsec-technique}
Let $A$ be a compact metric space with metric $d_A$, endowed with its Borel $\sigma$-algebra $\cB(A)$. In this subsection, we introduce two equivalent metrics on the space of probability measures $\cM(A)$. Using these metrics, we show that the distance between the sequence of societal summaries and the sequence of their realizations converges to zero in probability.

\begin{itemize}
	\item Let $\rho$ denote the \emph{Prohorov metric} on $\cM(A)$. That is, for all $\tau,\widetilde \tau \in \mathcal{M}(A)$, we have
	\[\rho( \tau, \widetilde \tau )=\inf \bigl\{\epsilon>0 \colon \tau(B) \leqslant \epsilon + \widetilde \tau (B^{\epsilon}) , \widetilde \tau(B) \leqslant \epsilon + \tau (B^{\epsilon}) \text { for all  } B \in \cB(A) \bigr\},\] where
	${B^\epsilon } = \bigl\{ a \in A \colon d_A(a,b) < \epsilon {\text{ for some }} b \in B \bigr\} .$
	\item Let $\beta$ denote the \emph{dual-bounded-Lipschitz metric} on $\mathcal{M}(A)$. That is, for all   $\tau,\widetilde \tau \in \mathcal{M}(A)$, we have
	\begin{equation*}
		\beta(\tau, \widetilde \tau) = \| \tau-\widetilde \tau\|_{{B L}}=\sup \Bigl\{\bigl|\int_{A} h {\mathrm{d}}(\tau- \widetilde \tau)\bigr| \colon\|h\|_{{B} {L}} \leqslant 1 \Bigr\},
	\end{equation*}
	where $h$ is bounded continuous on $A$,
	$\|h\|_{\infty} = \sup\limits_{a \in A}|h(a)|$, $\|h\|_{\mathrm{L}}=\sup\limits_{a \neq b, a,b \in A } \dfrac{|h(a)-h(b)|}{d_A(a, b)}$,  and $\|h\|_{{BL}}  = \|h\|_{\infty}+\|h\|_{{L}}.$
\end{itemize}

\smallskip
It is known in the literature that $\rho$ and $\beta$ are equivalent metrics; see, for example, \citet[Theorem 8.3.2]{Bo2007}. Recall that given a sequence $\{g^n\}_{n \in \mathbb{Z}_+}$, where each $g^n$ is a randomized strategy profile in the game $G_n$, there exists a probability space $(\Omega, \Sigma, \mathbb{P})$ and a sequence of random variables $\{x^n_i\}_{i \in I_n,\, n \in \mathbb{Z}_+}$ mapping from $\Omega$ to $A$ such that:
\begin{itemize}
    \item[(i)] For each $i \in I_n$ and $n \in \mathbb{Z}_+$, the distribution of $x^n_i$ equals $g^n(i)$;
    \item[(ii)] For each $n \in \mathbb{Z}_+$, the random variables $\{x^n_i\}_{i \in I_n}$ are pairwise independent.
\end{itemize}

\begin{lem}\label{lem-norm}
Let $\{G_n\}_{n \in \mathbb{Z}_+}$ be a sequence of finite-player games, and let $g^n$ be a randomized strategy profile of $G_n$ for each $n \in \mathbb{Z}_+$. For each $\omega \in \Omega$, let $s(x^n)(\omega) = \sum_{i \in I_n} \delta_{x^n_i(\omega)} \lambda_n(i)$ denote the realized societal summary associated with the strategy profile $g^n$. Thus $s(x^n)$ defines a random variable from $(\Omega, \Sigma, \mathbb{P})$ to $\cM(A)$. Then,
\[
\beta\bigl(s(x^n), s(g^n)\bigr) \to 0 \quad \text{and} \quad \rho\bigl(s(x^n), s(g^n)\bigr) \to 0 \quad \text{in probability,}
\]
where $s(g^n) = \int_{I_n} g^n(i)\, \mathrm{d}\lambda_n(i)$.
\end{lem}

\begin{proof}[Proof of Lemma~\ref{lem-norm}]
	We divide the proof into two steps. In step 1, we show that  for any bounded continuous function $h \colon A \to \bR$ with $||h||_{BL} \le 1$, $\int_{A} h {\rm{d}} \bigl( {s(x^n) - s(g^n)} \bigr) \to 0$ in probability. In step 2, we show that $\beta \bigl(s(x^n), s(g^n) \bigr) \to 0$ in probability. Finally, by the equivalence of $\rho$ and $\beta$, we obtain that $\rho \bigl(s(x^n), s(g^n) \bigr) \to 0$ in probability.

\noindent\textbf{Step 1.} In this step, we prove that
	for any  bounded and continuous function $h \colon A \to \bR$ with $\|h\|_{{BL}} \leq 1$, we have
	\begin{align}\label{align-converge}
		\int_{A} h {\rm{d}} \bigl( s(x^n) - s(g^n) \bigr) = \sum\limits_{i \in I_n}\lambda_{n}(i) \Bigl(h\bigl(x^n_i \bigr) - \bE\bigl[h\bigl(x^n_i \bigr)\bigr]\Bigr)  \to 0
	\end{align}
	in probability.
Fix any $n \in \mathbb{Z}_{+}$, since $\{x^n_i\}_{ i \in I_n}$ are pairwise independent and $h$ is a bounded continuous function, we know that $\{h ( x^n_i ) \}_{ i \in I_{n} }$ are also pairwise independent. By the definition of  $\|h\|_{BL} \leq 1$, we have $\|h\|_{\infty} \leq 1$ and hence $-1 \leq h (x^n_i) \leq 1$, $ \text{var}\bigl(h ( x^n_i)\bigr) \leq 1$, for all $i \in I_n$. Moreover, by the independence of $\{h ( x^n_i) \}_{ i \in I_{n} }$, we have
	$$\text{var} \Bigl( \sum\limits_{i \in I_n} h (x^n_i)\lambda_n(i) \Bigr) = \sum\limits_{i \in I_n} \bigl(\lambda_n(i) \bigr)^{2} \text{var} \bigl( h(x^n_i) \bigr).$$
   Since $\bE \bigl[ \sum\limits_{i \in I_n}\lambda_n(i) h (x^n_i) \bigr] = \sum\limits_{i \in I_n}\lambda_{n}(i)  \bE \bigl[ h(x^n_i )\bigr]$, for any $\varepsilon > 0$, we have
	\begin{align}\label{Lemma-Chebyshev}
			&   \mathbb P \Bigl(
			\Bigl|\sum_{i \in I_n}\lambda_n(i) \Bigl( h\bigl( x^n_i\bigr)-\bE \bigl[ h \bigl(x^n_i\bigr) \big] \Bigr)
			\Bigr| \leq \varepsilon \Bigr) \notag \\
			&\geqslant 1 - \frac{\sum\limits_{i \in I_n}  \bigl(\lambda_n(i)\bigr)^2 \text{var} \Bigl( h \bigl(x^n_i\bigr)  \Bigr)}{\varepsilon^2} \notag \\
			&\geqslant 1 -  \frac{\sum\limits_{i \in I_n} \bigl( \lambda_n(i)\bigr)^2  }{\varepsilon^2} \notag \\
			&\geqslant 1 -  \frac{\mathop {\sup }\limits_{j \in I_n} \lambda_n(j)\sum\limits_{i \in I_n} {\lambda_n(i)}  }{\varepsilon^2} \notag \\
			&\geqslant 1 -  \frac{\mathop {\sup }\limits_{j \in I_n} \lambda_n(j)  }{\varepsilon^2},
	\end{align}
	where the first inequality is due to the Chebyshev's inequality, and the last inequality follows from the fact that $\sum_{i \in I_n} {\lambda_{n}(i)} = 1 $. Combining with   $\sup_{j \in I_{n}} \lambda_{n}(j) \to 0 $ as $n \to \infty$, we can finish the proof of Formula~(\ref{align-converge}).
\smallskip

\noindent\textbf{Step 2.} In this step, we prove  that $\beta\bigl(s(x^n), s(g^n)\bigr) \to 0$ in probability. According to the proof in step 1, we know that for any finite number $m$, and a sequence of bounded continuous functions $\{ {p_{l}}\} _{l = 1}^m$ with $\|p_l\|_{BL} \leq 1$ for all $1 \le l \le m$,  we have
	\begin{align}\label{align-converge n}
		\sum\limits_{i \in I_n}\lambda_{n}(i) \Bigl( p_l (x^n_i) - \bE \bigl[ p_l( x^n_i) \bigr] \Bigr) \rightarrow 0
	\end{align}
	uniformly in probability for $l \in \{ 1,2,3, \cdots,m \}$.
	
	Let $ E = \{ h\colon \|h\|_{BL}\leq 1 \}$ be a compact space of bounded continuous functions. Given any $\varepsilon > 0$, there exists a finite number $m(\varepsilon)$, and a set of functions denoted by $\{ h_l  \}_{l = 1}^{m(\varepsilon)}$ such that
	\begin{itemize}
		\item[(i)] $h_1,h_2,...,h_{m(\varepsilon)} \in  E$,
		\item[(ii)] for any $h \in E$,  $  \inf\limits_{1 \le l \le m(\varepsilon)   } \sup\limits_{a \in A} \bigl| {h(a) - h_{l}(a)} \bigr| < \varepsilon $.
	\end{itemize}

	For any $h \in E$, we have
	\begin{align}\label{inequa-converge in prob}
			& \quad \Bigl|\int_{A} h \mathrm{d}\bigl(s(x^n)-s(g^n) \bigr)\Bigr|  \notag \\
			&\le   \inf\limits_{1 \le l  \le m(\varepsilon)} \left\{
			\Bigl| \int_{A} h_{l} \mathrm{d} \bigl(s(x^n)-s(g^n)\bigr)\Bigr| +
			\Bigl|\int_{A} (h - h_l) \mathrm{d}\bigl(s(x^n)-s(g^n)\bigr) \Bigr| \right\} \notag \\
			&\le \sup\limits_{1 \le l  \le m(\varepsilon)}
			\Bigl|\int_{A} h_{l} \mathrm{d} \bigl(s(x^n)-s(g^n)\bigr)\Bigr| + \inf\limits_{1 \le l  \le m(\varepsilon)}
			\Bigl|\int_{A} (h - h_{l}) \mathrm{d} \bigl(s(x^n)-s(g^n)\bigr) \Bigr| \notag \\
			& \le  \sup\limits_{1 \le l  \le m(\varepsilon)} \Bigl|\int_{A} h_{l} \mathrm{d} \bigl(s(x^n)-s(g^n)\bigr) \Bigr| + 2 \varepsilon,
	\end{align}
	where the first inequality follows from the triangle inequality, and the last inequality follows by   $  \inf_{1 \le l \le m(\varepsilon)   } \sup_{a \in A} | h(a) - h_l(a) | < \varepsilon $. Therefore,
	\[
	\sup\limits_{h \in E}
	\Bigl|\int_{A} h \mathrm{d} \bigl(s(x^n)-s(g^n) \bigr) \Bigr|
	\leq
	\sup\limits_{1 \le l \le m(\varepsilon) } \Bigl| {\int_{A} { h_{l} \mathrm{d} \bigl(s(x^n) - s(g^n) \bigr)} } \Bigr| + 2\varepsilon.
	\]
	To finish the proof of $\beta (s(x^n), s(g^n) ) \to 0$ in probability, it suffices to show that for any $\eta > 0$,
	\[
	\lim\limits_{n \rightarrow \infty} \mathbb P \Bigl(
	\beta \bigl( s(x^n), s(g^n) \bigr) \ge \eta
	\Bigr) =  \lim\limits_{n \rightarrow \infty} \mathbb P \Bigl(
	\sup\limits_{h \in E} \Bigl|\int_A h \mathrm{d} \bigl(s(x^n) - s(g^n) \bigr)\Bigr|
	\geq \eta
	\Bigr)
	= 0.\]
	Pick an $\varepsilon$ such that $0 < \varepsilon < \frac{\eta}{2}$, then we only need to show that
	\[
	\lim\limits_{n \rightarrow \infty}
	\mathbb P \Bigl(
	\sup\limits_{ 1 \le l \le m(\varepsilon) } \Bigl| \int_A h_l \mathrm{d} \bigl(s(x^n) - s(g^n) \bigr)  \Bigr| \geq  - 2\varepsilon + \eta
	\Bigr)
	= 0. \]
	Since
	\[ \mathbb P \Bigl(
	\sup\limits_{ 1 \le l \le m(\varepsilon) } \Bigl| \int_A h_l\mathrm{d} \bigl(s(x^n) - s(g^n) \bigr)  \Bigr| \geq  - 2\varepsilon + \eta
	\Bigr)
	\le
	\sum_{  l = 1  }^{ m(\varepsilon)}  \mathbb P \Bigl(
	\Bigl| \int_A h_l \mathrm{d} \bigl(s(x^n) - s(g^n)\bigr) \Bigr| \geq - 2\varepsilon +\eta
	\Bigr),
	\]
and by using Formula~(\ref{align-converge n}), we conclude that
	$\mathop {\sup }\limits_{h \in E} \Bigl|\int_A h \text{d} \bigl( s(x^n) - s(g^n)\bigr) \Bigr| \to 0$ in probability.
\end{proof}

\subsection{Proof of Lemma~\ref{lem:auxiliary}}\label{subsec-proof auxiliary}
Let $s(G,g) = \int_I \delta_{G(i)} \otimes g(i) \rmd \lambda(i)$ be the joint distribution of $G$ and $g$ on the product space $\cU_A \times A$. Clearly, the marginal distribution $s(G,g)|_{\cU_A} = \lambda G^{-1}$. Since both $\cU_A$ and $A$ are Polish spaces, there exists a family of Borel probability measures $\{\cS(u,\cdot) \}_{u \in \cU_A}$ in $\cM(A)$, which represents the disintegration of $s(G,g)$ with respect to $\lambda G^{-1}$ on $\cU_A$. Let $\overline g \colon  \cU_A \to \cM(A)$ be defined such that $\overline g  (u)  = \cS(u)$ for all  $u \in \cU_A$. According to \citet[Lemma 5]{SSY2020}, the composition mapping $\widetilde{g} = \overline{g}\circ G$ is a Nash equilibrium of $G$ that satisfies $s(\widetilde{g}) = s(g) = \tau$. Furthermore, we can ensure that for any player $i \in I$, her strategy $\widetilde{g}(i)$ is a best response with respect to the society summary  $s(\widetilde{g})$. This requirement can be met by modifying the strategies of a subset of players with measure zero.

\subsection{Proof of Theorem~\ref{thm:main1}}\label{subsec-proof lower}

The major difficulty of this proof is to estimate the difference between $u_i^n(g^n)$ and $u_i^n(\mu, g_{-i}^n)$ for all $\mu \in \cM(A)$, which we achieve by dividing the estimation into five steps: Step 1 restricts attention to uniformly bounded and equicontinuous payoff functions; Step 2 bounds the difference between $u_i^n(\mu, g_{-i}^n)$ and $u_i^n\bigl(\mu, s(\mu, g_{-i}^n)\bigr)$ for all $\mu \in \cM(A)$; Step 3 proves $\rho(s(g^n), s(\widetilde{g})) \to 0$ and combines this with Step 2 to bound the gap between $u_i^n(g^n)$ and $u_i^n\bigl(g_i^n, s(\widetilde{g})\bigr)$; Step 4 bounds the difference between $u_i^n(\mu, g_{-i}^n)$ and $u_i^n\bigl(\mu, s(\widetilde{g})\bigr)$ for all $\mu \in \cM(A)$ by leveraging Step 3 and the fact that $\rho(s(\mu, g_{-i}^n), s(\widetilde{g})) \to 0$; and finally, Step 5 uses the triangle inequality, the equilibrium property (since $s(\widetilde{g}) =  \tau^* = s(g)$, $g$ is an equilibrium of $G$ and the construction of $g^n$ ensures  the existence of $i'$ such that $g^n(i) = g(i')$, $u_i^n = u_{i'}$, hence $u_i^n(g^n(i), s(\widetilde{g})) \ge u_i^n(\mu, s(\widetilde{g}))$ for any $\mu \in \cM(A_i)$), and the results from Steps 3 -- 4 to complete the proof.

\smallskip

\noindent\textbf{Step 1.} We show that for any $\varepsilon > 0$, there exist a  sequence of subsets  $S_n \subseteq I_n$  such that $ \lambda_{n}(S_n) > 1- \frac{\varepsilon}{2}$ for all $n \in \mathbb{Z}_+$, and $\{ u_i^n  \}_{ i \in S_n, n \in \mathbb{Z}_+}$ are equicontinuous and uniformly bounded by a constant $M_\varepsilon$. For simplicity, let $ \mathcal{W}^{n} = \lambda_n  {G_n}^{-1} $, and $\mathcal{W} = \lambda  G^{-1}$. Since $A \times \cM(A)$ is a compact metric space, the space of bounded and  continuous functions $\cU_A$ on $A \times \cM(A)$ is a Polish space. By using the Prohorov theorem (\citet[Theorem 5.2]{B1999}), we know that $\{ \mathcal{W}^n \}_{n \in \bZ_{+}}$ is tight, which means that for any $\varepsilon > 0$, there exists a compact set $K_{\varepsilon} \subset \cU_A$ such that $\mathcal{W}^{n}(K_{\varepsilon}) > 1- \frac{\varepsilon}{2}   $ for all $n \in \bZ_{+} $.

Since $K_{\varepsilon}$ is a compact set that consists of bounded and continuous functions, the Arzelà-Ascoli theorem (\citet[Theorem 45.4]{Mu2000}) implies that all the functions in $K_{\varepsilon}$ are  uniformly bounded and equicontinuous functions. Let $M_\varepsilon$ denote a bound of all the functions in $K_{\varepsilon}$, and $S_n = \{ i \in I_n |  u_i^n \in K_{\varepsilon}  \}$ for all $n \in \mathbb{Z}_+$. It is clear that $\lambda_n(S_n ) > 1  - \frac{\varepsilon}{2}$.

\smallskip

\noindent\textbf{Step 2.}  We estimate $|u_i^n( \mu, g^n_{-i}) - u^n_i ( \mu , s(\mu, g^n_{-i}) )|$ in this step, where  $  u^n_i ( \mu , s(\mu, g^n_{-i}) ) = \int_{A_i} u^n_i ( a_i , \lambda_n(i) \mu + \sum_{j \in I_n \backslash \{i\} } \lambda_n(j) g^n (j) ) \mu(\rmd a_i)$. To be precise, we prove that for any $\varepsilon > 0$ and any sequence of randomized strategy profiles $\{g^n\}_{n \in \mathbb{Z}_+}$, there exists $ N_{\varepsilon} \in \mathbb Z_{+}$ such that for all $n \ge N_{\varepsilon}$, $i \in S_n$,  $\mu \in \cM(A)$, we have
			\[
			\left|	u_i^n( \mu, g^n_{-i}) - u^n_i \Bigl( \mu , \lambda_n(i) \mu + \sum\limits_{j \in I_n \backslash \{i\} } \lambda_n(j) g^n (j) \Bigr) \right| \le \frac{\varepsilon}{4}.
			\]

Similar to the proof of Lemma~\ref{lem-norm}, we represent the  strategies of all players in the sequence of games $\{G_n\}_{n \in \bZ_+}$  using random variables $\{ x^n_i\}_{i \in I_n, n \in \bZ_{+}}$, where each  $\{ x^n_i(\omega)\}_{i \in I_n, n \in \bZ_{+}}$  is a realization of a corresponding random variable  $ \{g_i^n\}_{i \in I_n}$. 
	Let $ x_{\mu}$ be a random variable that induces the distribution $\mu$, then we have
$$
u_i^n( \mu, g^n_{-i}) = \bE \Bigl[   u_i^n \Bigl(   x_{\mu} , \lambda_n(i) x_{\mu} +  \sum\limits_{j \in I_n \backslash \{i\} } \lambda_n(j)x^n_j       \Bigr) \Bigr],$$
and
$$u^n_i \Bigl( \mu , \lambda_n(i) \mu +  \sum\limits_{j \in I_n \backslash \{i\}} \lambda_n(j)g^n (j) \Bigr)
=
\bE \Bigl[   u_i^n \Bigl(   x_{\mu},  \lambda_n(i) \mu +  \sum\limits_{j \in I_n \backslash \{i\}} \lambda_n(j)g^n (j)      \Bigr) \Bigr].$$
Hereafter, we restrict our attention to functions  $ u, u_i^n \in K_{\varepsilon}$.
By uniformly equicontinuity, we know that for any $\varepsilon > 0$, there exists $\eta>0$ such that for any $ \tau, \widetilde{\tau} \in \cM(A)$  with $\rho( \tau, \widetilde \tau) \le \eta$, 
\begin{equation}\label{equa-payoff equi}
	\begin{split}
		\bigl|  u ( a , \tau )   - u ( a , \widetilde \tau)              \bigr|
		\le
		\frac{\varepsilon}{4(2M_\varepsilon+1)}.
	\end{split}
\end{equation}

Let $s (x_{\mu}, x^n_{-i})   (\omega) =  \lambda_n(i) \delta_{x_{\mu}(\omega)} +  \sum_{j \in I_n \backslash \{i\}   } \lambda_{n}(j)  \delta _{
	x^n_j(\omega)}  $, and $s (\mu, g^n_{-i})    =  \lambda_n(i) \mu +  \sum_{j \in I_n \backslash \{i\}  } \\ \lambda_{n}(j) g^n(j)  $, for all $\omega \in \Omega$, $\mu \in \cM(A)$.
The triangle inequality implies that,
\begin{align}\label{align-metric triangle}
	\rho \bigl( s (\mu, g^n_{-i})   ,  s (x_{\mu}, x^n_{-i})   (\omega)  \bigr)   \le  &\rho \bigl(  s(g^n) ,  s (\mu, g^n_{-i})    \bigr) \notag \\
	 & + \rho \bigl(   s(x^n)(\omega),  s(g^n)  \bigr) \notag  \\
&	  	+
	\rho \bigl(  s(x^n)(\omega)  ,  s (x_{\mu}, x^n_{-i})   (\omega)   \bigr).
\end{align}
By the definition of Prohorov metric $\rho$, we know that for any $\omega \in \Omega$, $\mu \in \cM(A)$, $i \in I_n$,
\begin{align*}
	\rho \bigl( s(x^n)(\omega)   ,  s (x_{\mu}, x^n_{-i})   (\omega)    \bigr) \le  \sup\limits_{j \in I_{n}} \lambda_{n}(j).
\end{align*}
Since $\sup\limits_{j \in I_{n}} \lambda_{n}(j) \rightarrow 0 $, there exists $N_1 \in \mathbb{Z}_+$ such that for any $ n \ge N_1$, we have $\mathop{\sup}\limits_{ j\in I_{n}} \lambda_n(j) < \frac{\eta}{4} $. Hence for any $n \ge N_1$, $i \in I_n$, $\mu \in \cM(A)$, $\omega \in \Omega$,
\begin{align}\label{equa-metric gap 1}
	\rho \bigl( s(x^n)(\omega)   ,  s(x_{\mu}, x^n_{-i})   (\omega)    \bigr)  < \frac{\eta}{4}.
\end{align}
By the same argument as above, we can see that for any $n \ge N_1$, $i \in I_n$, $\mu \in \cM(A)$,
\begin{align}\label{equa-metric gap 2}
	\rho \bigl(  s(g^n) ,   s (\mu, g^n_{-i}) \bigr)   < \frac{\eta}{4}.
\end{align}

Let  $\Omega_1^{(\frac{\eta}{2},n)} = \bigl\{   \omega \in \Omega  \, | \,  \rho \bigl(  s(x^n)(\omega),  s(g^n)  \bigr) 	 < \frac{\eta}{2}    \} $ and $\Omega_2^{(\frac{\eta}{2},n)} = \Omega \backslash \Omega_1^{(\frac{\eta}{2},n)}$.
By Lemma~\ref{lem-norm}, for any $\varepsilon > 0 $, there exists $N_{\varepsilon} \ge N_1$ such that for any $n \ge N_{\varepsilon}$,
\begin{align}\label{equa-payoff gap 1}
	\mathbb P \Bigl( \Omega_2^{(\frac{\eta}{2},n)} \Bigr)   \le   \frac{\varepsilon }{4(2M_\varepsilon + 1)}.
\end{align}
Let
$$ H_1^{(\frac{\eta}{2}, n)} =
\left| \bE \Bigl[ \Bigl( u_i^n \bigl(   x_{\mu} ,   s (\mu, g^n_{-i})    \bigr)
-
u^n_i \bigl( x_{\mu} , s(x_{\mu},x^n_{-i})  \bigr) \Bigr) \delta_{ \Omega_1^{(\frac{\eta}{2},n)}  }  \Bigr]
\right|,$$
and
$$ H_2^{(\frac{\eta}{2}, n)}
=
\left| \bE \Bigl[ \Bigl( u_i^n \bigl(   x_{\mu} ,   s (\mu, g^n_{-i})   \bigr)
-
u^n_i \bigl( x_{\mu} , s(x_{\mu},x^n_{-i}) \bigr) \Bigr) \delta_{ \Omega_2^{(\frac{\eta}{2},n)}  }  \Bigr]
\right|,$$
By using the triangle inequality, we have
\begin{equation*}
	\begin{split}
		\left| \bE \Bigl[  u_i^n \bigl(   x_{\mu} ,   s (\mu, g^n_{-i})    \bigr)
		-
		u^n_i \bigl( x_{\mu} , s( x_{\mu},x^n_{-i}) \bigr)   \Bigr]
		\right|
		\leq H_1^{(\frac{\eta}{2}, n)} + H_2^{(\frac{\eta}{2}, n)}.
	\end{split}
\end{equation*}

Then we estimate $H_1^{(\frac{\eta}{2}, n)}$ and $H_2^{(\frac{\eta}{2}, n)}$ separately.  Note that for any $n \ge N_{\varepsilon}$ and any player $i \in S_n$, we have $u_i^n \in K_{\varepsilon}$.
\begin{itemize}
	\item[(i)] By the definition of event $\Omega_1^{(\frac{\eta}{2}, n)}$ and  Inequalities~(\ref{equa-payoff equi}), (\ref{align-metric triangle}), (\ref{equa-metric gap 1}), and (\ref{equa-metric gap 2}), we can see that
	\begin{align}\label{equa-payoff-H1}
		H_1^{(\frac{\eta}{2}, n)} \le \frac{\varepsilon }{4(2M_\varepsilon + 1)}.
	\end{align}	
	\item[(ii)] Since $u_i^n$ is bounded by $M_\varepsilon$, combined with Inequality~(\ref{equa-payoff gap 1}) we have
	\begin{align}\label{equa-payoff-H2}
		H_2^{(\frac{\eta}{2}, n)} \le  2M_\varepsilon\frac{\varepsilon }{4(2M_\varepsilon + 1)}.
	\end{align}
\end{itemize}
Combining Inequalities (\ref{equa-payoff-H1}) and (\ref{equa-payoff-H2}), for any $n \ge N_{\varepsilon}$, we have
\begin{align*}
	\left|
	\bE \Bigl[u^n_i \bigl( x_{\mu}, \lambda_n(i) \mu + \sum\limits_{j \in I_n \backslash \{i\}} \lambda_n(j)g^n(j) \bigr)\Bigr] - \bE \Bigl[u^n_i \bigl( x_{\mu} , \lambda_n(i) \delta_{ x_{\mu} } + \sum\limits_{j  \in I_n \backslash \{i\}} \lambda_{n}(j)\delta_{ x^n_j } \bigr)\Bigr]
	\right|
	\le \frac{\varepsilon}{4}.
\end{align*}
That is,
$$
\left|	u_i^n( \mu, g^n_{-i}) - u^n_i \Bigl( \mu , \lambda_n(i) \mu + \sum\limits_{j \in I_n \backslash \{i\} } \lambda_n(j) g^n (j) \Bigr) \right| \le \frac{\varepsilon}{4}.
$$

\smallskip

\noindent\textbf{Step 3.} In this step, we estimate  $\Bigl|u_i^n(g^n) - u_i^n \bigl(g^n_i, s(\widetilde g)\bigr)\Bigr|$. Similarly, let $\{ x^n_i \}_{i \in I_n, n \in \bZ_{+}}$ represent all the players' strategies in games $\{G_n\}_{n \in \bZ_+}$. The payoff of player $i \in I_n$ in game $G_n$ with the strategy profile $g^n$ can be rewritten as
\[
u_i^n(g^n) = \bE \Bigl[  u_i^n \bigl( x^n_i,   \sum\limits_{j \in I_n} \lambda_n(j) \delta_{x^n_j}   \bigr)    \Bigr].
\]
Since there exists some player $i' \in I$ such that $G(i') = G_n(i)$ (i.e., $u_{i'} = u^n_i$), the payoff of player $i'$ in large game $G$  with strategy profile $\widetilde g$ is equivalent to
\[
\int_A u_{i'} \bigl(  a, s(\widetilde g) \bigr)  \widetilde g( i', \rmd a) = \int_A u_i^n \bigl(  a, s(\widetilde g) \bigr)  g^n( i, \rmd a) = \bE \bigl[  u_i^n \bigl( x^n_i,   s(\widetilde g) \bigr)    \bigr].
\]
By the triangle inequality, we have
\begin{align*}
	&   \left|  \bE \Bigl[  u_i^n \bigl( x^n_i,   \sum\limits_{j \in I_n} \lambda_n(j) \delta_{x^n_j}   \bigr)    \Bigr]
	-
	\bE \Bigl[  u_i^n \bigl( x^n_i,   s(\widetilde g) \bigr)    \Bigr]  \right|   \\
	&   \le
	\left| \bE \Bigl[  u_i^n \bigl( x^n_i,   \sum\limits_{j \in I_n} \lambda_n(j) \delta_{x^n_j}   \bigr)    \Bigr] -  \bE \Bigl[  u_i^n \bigl( x^n_i,   \sum\limits_{j \in I} \lambda_n(j) g^n(j)   \bigr)    \Bigr]  \right|    \tag{i}   \\
	&  +    \left|  \bE \Bigl[  u_i^n \bigl( x^n_i,   \sum\limits_{j \in I_n} \lambda_n(j) g^n(j)   \bigr)    \Bigr] -  \bE \Bigl[  u_i^n \bigl( x^n_i,   s(\widetilde g) \bigr)    \Bigr]  \right| .    \tag{ii}
\end{align*}

\noindent Estimation of (i). By step 2, for any $\varepsilon > 0$, there exists a sequence of sets $\{S_n\}_{n \in \mathbb{Z}_+}$ such that $\lambda_n(  S_n  ) > 1 - \frac{\varepsilon}{2}$ for all $n \in \mathbb{Z}_+$, and a number $ N_{\varepsilon} \in \mathbb{Z}_+$ such that for all $n \ge  N_{\varepsilon}$, $i \in S_n$, we have $\Bigl|u_i^n(g^n) - u_i^n\bigl(g^n(i),\sum_{j \in I_n} \lambda_n(j) g^n(j)\bigr)\Bigr| < \frac{\varepsilon}{4}$. That is, 
$$\left| \bE \Bigl[  u_i^n \bigl( x^n_i,   \sum\limits_{j \in I_n} \lambda_n(j) \delta_{x^n_j}   \bigr)    \Bigr] -  \bE \Bigl[  u_i^n \bigl( x^n_i,   \sum\limits_{j \in I} \lambda_n(j) g^n(j)   \bigr)    \Bigr]  \right|  <  \frac{\varepsilon}{4}.$$

\noindent Estimation of (ii). Since all the functions in $\{u_i^n\}_{i \in S_n, n \in \mathbb{Z}_+}$ are uniformly bounded by $M_\varepsilon$ and equicontinuous, for the given $\varepsilon > 0$, there exists $\eta > 0$ such that for any $\tau, \widetilde \tau \in \cM(A)$ with $\rho(\tau, \widetilde \tau) < \eta$, we have  $|u(a,\tau) - u(a, \widetilde \tau)| < \frac{\varepsilon}{4( 2M_\varepsilon  + 1)}$, for all $a \in A, u \in \{u_i^n\}_{i \in S_n, n \in \mathbb{Z}_+}$. For any bounded continuous function $h  \colon A \to R$, let  $\phi(u) = \int_{A} h(a)  \overline g ( u , \rmd a)$ for all $u \in \mathrm{supp } \lambda G^{-1} \subset \cU_A$. Since $\overline g$ is continuous for $\lambda G^{-1}$-almost all $u \in \mathrm{supp } \lambda G^{-1}$, we know that $\phi$ is bounded and continuous for $\lambda G^{-1}$-almost all $u \in \mathrm{supp } \lambda G^{-1}$.  Moreover, as $\lambda_n {G_n}^{-1}$ converges weakly to $\lambda G^{-1}$ and $\mathrm{supp } \lambda_n {G_n}^{-1} \subset \mathrm{supp } \lambda G^{-1}$, by Portmanteau Theorem (\citet[Theorem 13.16]{Kl2014}), we have that 
$$\int_{\mathrm{supp } \lambda G^{-1}} \phi(u) \rmd \lambda_n{G_n}^{-1} (u) \to \int_{\mathrm{supp } \lambda G^{-1}} \phi(u) \rmd \lambda G^{-1} (u).$$
By changing of variables, we can see that $$\int_{\mathrm{supp } \lambda G^{-1}} \phi(u) \rmd \lambda_n{G_n}^{-1} (u) =
\int_{I_n} \phi \bigl( G_n(i) \bigr) \rmd  \lambda_n (i)
,$$
and
$$\int_{\mathrm{supp } \lambda G^{-1}} \phi(u) \rmd \lambda G^{-1} (u)
=
\int_I  \phi \bigl( G(i) \bigr)  \rmd  \lambda (i).
$$
According to the definitions of $\overline g$, $\widetilde g$, and $g^n$, we know that
$$ \phi \bigl( G_n(i) \bigr) = \int_A h(a) \overline g \bigl( G_n(i), \rmd a \bigr)  = \int_A h(a)  g^n(i, \rmd a),$$
and
$$ \phi \bigl( G(i) \bigr) = \int_A h(a) \overline g \bigl( G(i), \rmd a \bigr)  = \int_A h(a)  \widetilde g(i, \rmd a).$$
Hence for any continuous function $h$, we have 
$$\int_{I_n} \int_A h(a)  g^n(i, \rmd a)  \rmd \lambda_n(i) \to \int_I \int_A h(a) \widetilde g(i, \rmd a) \lambda (i),$$
which is equivalent to 
$$\int_A  h(a) s(g^n)(\rmd a) \to \int_A  h(a) s(\widetilde g)(\rmd a)$$
by changing of variables. Therefore, $s(g^n)$ converges weakly to $s(\widetilde g)$, and there exists $\widetilde N_1 \in \bZ_{+}$  such that for all $n \ge \widetilde N_1$,
$\rho\bigl( s(g^n) , s(\widetilde g) \bigr) < \eta$.  Thus for all $n \ge \widetilde N_1$, $i \in S_n$, we have
\[
\left|  \bE \Bigl[  u_i^n \bigl( x^n_i,   \sum\limits_{j \in I_n} \lambda_n(j) g^n(j)    \bigr)    \Bigr] -  \bE \Bigl[  u_i^n \bigl( x^n_i,   s(\widetilde g)  \bigr)    \Bigr]  \right|
\le
\dfrac{\varepsilon}{4( 2M_\varepsilon + 1  )}.
\]
Combine the estimations of part (i) and part (ii) above, we have that
\[
\left|  \bE \Bigl[  u_i^n \Bigl( x^n_i,   \sum\limits_{j \in I_n} \lambda_n(j) \delta_{x^n_j}   \Bigr)    \Bigr]
-
\bE \Bigl[  u_i^n \bigl( x^n_i,   s(\widetilde g) \bigr)    \Bigr]  \right|  \le \frac{\varepsilon}{4} + \dfrac{\varepsilon}{4( 2M_\varepsilon + 1  )},
\]
for all $n \ge \max\{  \widetilde N_1,  N_{\varepsilon}  \}$, $i \in S_n$. Therefore, for all $n \ge \max\{  \widetilde N_1,  N_{\varepsilon}  \}$, $i \in S_n$, we have
\[
\Bigl| u_i^n(g^n) - \int_{A} u_i^n \bigl( a, s(\widetilde g) \bigr) g^n(i, \rmd a) \Bigr| \le \frac{\varepsilon}{4} + \dfrac{\varepsilon}{4( 2M_\varepsilon + 1  )}.
\]
That is, 
$$
\Bigl|u_i^n(g^n) - u_i^n \bigl(g^n_i, s(\widetilde g)\bigr)\Bigr| \le \frac{\varepsilon}{4} + \dfrac{\varepsilon}{4( 2M_\varepsilon + 1  )}.
$$

\smallskip

\noindent \textbf{Step 4.} We estimate $\Bigl|u_i^n ( \mu, g_{-i}^n) - u_i^n \bigl(\mu, s(\widetilde g)\bigr)\Bigr|$ in this step. By using the triangle inequality, we have that
\begin{align*}
	&   \left|  \bE \Bigl[  u_i^n \bigl(  x_{\mu},   \lambda_n(i) \delta_{ x_{\mu}} + \sum\limits_{j \in I_n \backslash \{i\} } \lambda_n(j) \delta_{x^n_i}   \bigr)    \Bigr] -  \bE \Bigl[  u_i^n \bigl( x_{\mu},   s(\widetilde g) \bigr)    \Bigr]  \right|   \\
	&   \le
	\left|  \bE \Bigl[  u_i^n \bigl( x_{\mu},   \lambda_n(i) \delta_{ x_{\mu}} + \sum\limits_{j \in I_n \backslash \{i\} } \lambda_n(j) \delta_{x^n_j}   \bigr)    \Bigr] -  \bE \Bigl[  u_i^n \bigl(  x_{\mu},   \lambda_n(i) \mu + \sum\limits_{j \in I_n \backslash \{i\} } \lambda_n(j) g^n(j)   \bigr)     \Bigr]  \right|    \tag{I}   \\
	& + \left|   \bE \Bigl[  u_i^n \bigl(  x_{\mu},   \lambda_n(i) \mu + \sum\limits_{j \in I_n \backslash \{i\} } \lambda_n(j) g^n(j)   \bigr)     \Bigr]   - \bE \Bigl[  u_i^n \bigl( x_{\mu},   \sum\limits_{j \in I_n} \lambda_n(j) g^n(j)    \bigr)    \Bigr]  \right|    \tag{II}    \\
	& +    \left|  \bE \Bigl[  u_i^n \bigl( x_{\mu},   \sum\limits_{j \in I_n} \lambda_n(j) g^n(j)    \bigr)    \Bigr] -  \bE \Bigl[  u_i^n \bigl(  x_{\mu},   s(\widetilde g)  \bigr)    \Bigr]  \right|     \tag{III}
\end{align*}
for all $\mu \in \cM(A)$, where $ x_{\mu}$ is the random variable which induces the distribution $\mu$.
The estimation of part (I) is the same as part (i) in step 3, while the estimation of part (III) is the same as part (ii) in step 3.  Hence there exists $\widetilde  N_2 \in \bZ_{+} $ such that for all $n \ge \widetilde N_2$, $i \in S_n$, $\mu \in \cM(A)$,
\[
\text{(I) } 
+
\text{ (III) }
\le
\frac{\varepsilon}{4} + \dfrac{\varepsilon}{4( 2M_\varepsilon + 1  )}. 	
\]

Below we estimate part (II). By the definition of the Prohorov metric $\rho$, we have that for any  $\mu \in \cM(A)$, $i \in I_n$,
\begin{align*}
	\rho \bigl(  s(g^n) ,   s(\mu, g^n_{-i})  \bigr)  \le  \sup\limits_{j \in I_{n}} \lambda_{n}(j).
\end{align*}
Since $\sup_{j \in I_{n}} \lambda_{n}(j) \rightarrow 0 $, there exists $\widetilde N_3 \in \mathbb{Z}_{+}$ such that for any $ n \ge \widetilde N_3$, $\sup_{ j\in I_n} \lambda_n(j) < \eta $. Thus, for any $n \ge \widetilde N_3$, $i \in I_n$, $\mu \in \cM(A)$,
$
\rho \bigl(  s(g^n) ,   s (\mu, g^n_{-i}) \bigr)   < \eta.
$
Recall that for all $i \in S_n$, $u_i^n$ is uniformly bounded by $M_\varepsilon$ and equicontinuous. As we proved in part (ii), we have that 
\[
\left|   \bE \Bigl[  u_i^n \bigl(  x_{\mu},   \lambda_n(i) \mu + \sum\limits_{j \in I_n \backslash \{i\} } \lambda_n(j) g^n(j)   \bigr)     \Bigr]   - \bE \Bigl[  u_i^n \bigl( x_{\mu},   \sum\limits_{j \in I_n} \lambda_n(j) g^n(j)    \bigr)    \Bigr]  \right| \le  \dfrac{\varepsilon}{4( 2M_\varepsilon + 1  )},
\]
for all  $n \ge  \widetilde N_3  $, $i \in S_n$, $\mu \in \cM(A)$. Thus for all $n \ge  \max \{ \widetilde N_2, \widetilde N_3\}$, $i \in S_n$, $\mu \in \cM(A)$,
\[
\left|  \bE \Bigl[  u_i^n \bigl(  x_{\mu},   \lambda_n(i) \delta_{ x_{\mu}} + \sum\limits_{j \in I_n \backslash \{i\} } \lambda_n(j) \delta_{x^n_j}   \bigr)    \Bigr] -  \bE \Bigl[  u_i^n \bigl( x_{\mu},   s(\widetilde g) \bigr)    \Bigr]  \right|  \le \frac{\varepsilon}{4} + \dfrac{\varepsilon}{2( 2M_\varepsilon + 1  )}.
\]
That is, 
$$\Bigl|u_i^n ( \mu, g_{-i}^n) - u_i^n \bigl(\mu, s(\widetilde g)\bigr)\Bigr| \le \frac{\varepsilon}{4} + \dfrac{\varepsilon}{2( 2M_\varepsilon + 1  )}.
$$

\smallskip

\noindent \textbf{Step 5. } For any $n \geq \max \{ N_{\varepsilon}, \widetilde N_{1}, \widetilde N_{2}, \widetilde N_3\}$, $i \in S_n$, $\mu \in \cM(A)$, we have
\begin{align*}
	&\bE \Bigl[  u_i^n \bigl( x^n_i,   \sum\limits_{j \in I_n} \lambda_n(j) \delta_{x^n_j}   \bigr)    \Bigr] \\
	& \ge
	\bE \Bigl[  u_i^n \bigl( x^n_i,   s(\widetilde g) \bigr)    \Bigr] - \frac{\varepsilon }{4} - \frac{\varepsilon }{4(2M_\varepsilon + 1)} \\
	& \ge \bE \Bigl[ u_i^n \bigl(  x_{\mu} ,  s(\widetilde g)  \bigr) \Bigr]  - \frac{\varepsilon }{4} - \frac{\varepsilon }{4(2M_\varepsilon + 1)} \\
	&\ge  \bE \Bigl[  u_i^n \bigl( x_{\mu} ,  \lambda_n(i) \delta_{ x_{\mu}} + \sum\limits_{j \in I_n \backslash \{i\}} \lambda_n(j) \delta_{x^n_j}   \bigr)    \Bigr] - \frac{\varepsilon }{2} - \frac{\varepsilon }{2(2M_\varepsilon + 1)}.
\end{align*}
The first inequality follows from step 3 and the last inequality follows from step 4.  The second inequality is due to the fact that $\widetilde g$ is a Nash equilibrium of $G$ such that for any player $i \in I$, her strategy $\widetilde{g}(i)$ is a best response with respect to the society summary  $s(\widetilde{g})$, and $G_n(i) \in G(I)$ for all $i \in I_n$, $n \in \bZ_{+}$.
Hence we have that
\[
\bE \Bigl[  u_i^n \bigl( x^n_i,   \sum\limits_{j \in I_n} \lambda_n(j) \delta_{x^n_j}   \bigr)    \Bigr]
\ge
\bE \Bigl[  u_i^n \bigl( x_{\mu} ,  \lambda_n(i) \delta_{ x_{\mu}} + \sum\limits_{j \in I_n \backslash \{i\}} \lambda_n(j) \delta_{x^n_j}   \bigr)    \Bigr] - \varepsilon,
\]
for all  $n \ge \max\{ N_{\varepsilon}, \widetilde N_1, \widetilde N_2, \widetilde N_3\}$, $i \in S_n$, $\mu \in \cM(A)$.
Thus we conclude that $g^{n}$  is an $\varepsilon$-Nash equilibrium of $G_{n}$, for all $n \ge \max\{ N_{\varepsilon}, \widetilde N_1, \widetilde N_2, \widetilde N_3\}$.

Given any sequence  $\{\varepsilon_n\}_{n \in \bZ_{+}}$ such that $\varepsilon_n > 0$ for all $n \in \bZ_{+}$ and $\{\varepsilon_n\}_{n \in \bZ_{+}}$ converges to $0$, there exists a strictly increasing sequence  $\{\overline N_n\}_{n \in \bZ_{+}}$ such that $\overline N_n  \in \bZ_{+}$ for all $n \in \bZ_{+}$, and $g^{m}$ is a $\varepsilon_n$-Nash equilibrium of $G_{m}$ for all $m \ge \overline N_n $. Let $\varepsilon_k= \varepsilon_n$ if $ \overline N_n \le k < \overline N_{n+1}$, for all $n, k \in \bZ_{+}$. Thus we have that $\{ \varepsilon_k\}_{k \in \bZ_{+} }$ converges to $0$, and $g^k$ is an $\varepsilon_k$-Nash equilibrium of $G_k$ for all $k \ge \overline N_1$, which completes our proof of Theorem~\ref{thm:main1}.

\subsection{Proof of Theorem~\ref{thm:main2}}\label{subsec-proof ex post}
We need to show that, for any $\varepsilon > 0$ and any $\alpha>0$, there exists an integer $N$ such that for any $n \ge N$, the following holds:
$$\bP(\Omega^{n}_\varepsilon) \ge 1 - \alpha.$$
To simplify the analysis, we will adopt the same notation as in the previous proof. For each $\omega \in \Omega$, $x^n(\omega)$ represents a possible realization of $g^n$. Thus, to show that $x^n(\omega)$ is an approximate equilibrium, we only need to estimate the following difference:
\begin{align*}
&u_i^n\bigl( x^n(\omega) \bigr) - u^n_i\bigl(a, x^n_{-i}(\omega) \bigr)\\ 
= &u_i^n\bigl( x^n_i(\omega), s(x^n)(\omega) \bigr) - u^n_i\Bigl(a, s\bigl(\delta_a, x^n_{-i}(\omega)\bigr) \Bigr)
\end{align*}
where $a \in A$, and $s\bigl(\delta_a, x^n_{-i}(\omega)\bigr) = \lambda_n(i) \delta_a +  \sum_{j \in I_n \backslash \{i\}   } \lambda_{n}(j)  \delta _{x^n_j(\omega)}$.

Clearly, the difference between $u_i^n\bigl( x^n_i(\omega), s(x^n)(\omega) \bigr)$ and $u^n_i\Bigl(a, s\bigl(\delta_a, x^n_{-i}(\omega)\bigr) \Bigr)$ can be decomposed into three parts:
\begin{align*}
&u_i^n\bigl( x^n_i(\omega), s(x^n)(\omega) \bigr) - u^n_i\Bigl(a, s\bigl(\delta_a, x^n_{-i}(\omega)\bigr) \Bigr)\\
= &u_i^n\bigl( x^n_i(\omega), s(x^n)(\omega) \bigr) - u_i^n\bigl( x^n_i(\omega), s(g^n) \bigr) \\
&+ u_i^n\bigl( x^n_i(\omega), s(g^n) \bigr) - u_i^n\bigl(a, s(g^n) \bigr) \\
&+ u_i^n\bigl(a, s(g^n) \bigr) - u^n_i\Bigl(a, s\bigl(\delta_a, x^n_{-i}(\omega)\bigr) \Bigr).
\end{align*}
By Lemma~\ref{lem-norm}, $\rho \bigl( s(x^n), s(g^n) \bigr) \to 0$  in probability, which implies that there exists an integer $N_1$ and a  sequence of subset $\Omega^n \subset \Omega$ such that $\bP(\Omega^n) \ge 1-\alpha$, and for each $\omega \in \Omega^n$, $n \ge N_1$, $i \in S_n$, we have
$$u_i^n\bigl( x^n_i(\omega), s(x^n)(\omega) \bigr) - u_i^n\bigl( x^n_i(\omega), s(g^n) \bigr) \ge -\tfrac{\varepsilon}{4},$$
and
$$ u_i^n\bigl(a, s(g^n) \bigr) - u^n_i\Bigl(a, s\bigl(\delta_a, x^n_{-i}(\omega)\bigr) \Bigr) \ge -\tfrac{\varepsilon}{4}.$$

Thus, it suffices to estimate the term $u_i^n\bigl( x^n_i(\omega), s(g^n) \bigr) - u_i^n\bigl(a, s(g^n) \bigr)$. Similarly, it can be decomposed into the following three parts:
\begin{align*}
&u_i^n\bigl( x^n_i(\omega), s(g^n) \bigr) - u_i^n\bigl(a, s(g^n) \bigr) \\
= &u_i^n\bigl( x^n_i(\omega), s(g^n) \bigr) - u_i^n\bigl( x^n_i(\omega), g^n_{-i} \bigr) \\
&+ u_i^n\bigl( x^n_i(\omega), g^n_{-i} \bigr) - u^n_i(a, g^n_{-i}) \\
&+ u^n_i(a, g^n_{-i}) - u_i^n\bigl(a, s(g^n) \bigr).
\end{align*}
According to the Step 2 of the previous proof, there exists an integer $N_2$ such that for each $n \ge N_2$,  we have
$$u_i^n\bigl( x^n_i(\omega), s(g^n) \bigr) - u_i^n\bigl( x^n_i(\omega), g^n_{-i} \bigr) \ge -\tfrac{\varepsilon}{6}$$
and
$$u^n_i(a, g^n_{-i}) - u_i^n\bigl(a, s(g^n) \bigr) \ge -\tfrac{\varepsilon}{6}.$$

By Theorem~\ref{thm:main1}, we know that $g^n$ is an approximate equilibrium  of $G_n$. 
By using the approximate equilibrium property of $g^n$, there exists an integer $N_3$ such that 
$$u_i^n\bigl( x^n_i(\omega), g^n_{-i} \bigr) - u^n_i(a, g^n_{-i}) \ge -\tfrac{\varepsilon}{6}$$
for $n \ge N_3$. Therefore, for each $n \ge \max\{N_2, N_3\}$,
$$u_i^n\bigl( x^n_i(\omega), s(g^n) \bigr) - u_i^n\bigl(a, s(g^n) \bigr) \ge -\tfrac{\varepsilon}{2}.$$

In conclusion, let $N = \max\{N_1, N_2, N_3\}$. Then, for each $\omega \in \Omega^*$, we have
$$u_i^n\bigl( x^n(\omega) \bigr) - u^n_i\bigl(a, x^n_{-i}(\omega) \bigr) \ge -\tfrac{\varepsilon}{4} -\tfrac{\varepsilon}{2} -\tfrac{\varepsilon}{4} =  -\varepsilon,$$
which implies that $x^n(\omega)$ is an $\varepsilon$-Nash equilibrium. Therefore, 
$$\bP(\Omega^{n}_\varepsilon) \ge \bP(\Omega^n) \ge 1 - \alpha.$$


{\small
\singlespacing

\end{document}